\documentclass[journal]{IEEEtran}
\usepackage{amsthm}
\usepackage{graphicx}
\usepackage{subcaption}
\usepackage{bbm}

\newtheorem{theorem}{Theorem}

\newtheorem{assumption}{Assumption}

\usepackage{cite}

\usepackage{bm}

\usepackage{amsmath}
\usepackage{amssymb}
\usepackage{mathtools}
\DeclareMathOperator*{\argmin}{arg\,min}

\usepackage{algorithm}
\usepackage{algorithmic}
\makeatletter
\def\BState{\State\hskip-\ALG@thistlm}
\makeatother
\usepackage{fixltx2e}

\usepackage{stfloats}

\usepackage{amsthm}
\usepackage{xcolor}
\usepackage{lipsum}
\usepackage[margin=0.75in]{geometry}
\newgeometry{top=1in, left=0.70in,right=0.70in,bottom=0.70in}

\begin{document}
\setlength{\belowcaptionskip}{-10pt}
\setlength{\abovedisplayskip}{2pt}
\setlength{\belowdisplayskip}{2pt}
\setlength{\parskip}{0pt}
%
\title{Constrained Thompson Sampling for Real-Time Electricity Pricing with Grid Reliability Constraints}

\author{\IEEEauthorblockN{Nathaniel Tucker$^\dagger$ \quad}
\and
\IEEEauthorblockN{Ahmadreza Moradipari$^\dagger$\quad}
\and
\IEEEauthorblockN{Mahnoosh Alizadeh}
\\$^\dagger$Authors have equal contribution
\vspace*{-0.65cm}
}


%



\maketitle

\begin{abstract}
We consider the problem of an electricity aggregator attempting to learn customers' electricity usage models while implementing a load shaping program by means of broadcasting dispatch signals in real-time. We adopt a multi-armed bandit problem formulation to account for the stochastic and unknown nature of customers' responses to dispatch signals. We propose a constrained Thompson sampling heuristic, Con-TS-RTP, as a solution to the load shaping problem of the electricity aggregator attempting to influence customers' usage to match various desired demand profiles (i.e., to reduce demand at peak hours, integrate more intermittent renewable generation, track a desired daily load profile, etc). The proposed Con-TS-RTP heuristic accounts for day-varying target load profiles (i.e., multiple target load profiles reflecting renewable forecasts and desired demand patterns) and takes into account the operational constraints of a distribution system to ensure that customers receive adequate service and to avoid potential grid failures. We provide a discussion on the regret bounds for our algorithm as well as a discussion on the operational reliability of the distribution system's constraints being upheld throughout the learning process.
\end{abstract}


%
\IEEEpeerreviewmaketitle
\newtheorem{proposition}{Proposition}
\makeatletter
\def\blfootnote{\xdef\@thefnmark{}\@footnotetext}
\makeatother

\blfootnote{
\indent
This work was supported by NSF grants \#1847096 and \#1737565 and UCOP Grant LFR-18-548175.\\
\hspace{8.5pt}N. Tucker, A. Moradipari, and  M. Alizadeh are with the Department of Electrical and Computer Engineering, University of California, Santa Barbara, CA 93106 USA (email: nathaniel\_tucker@ucsb.edu)}

\section{Introduction}
\label{section: Intro}
In order to integrate the increasing volume of intermittent renewable generation in modern power grids,  aggregators are exploring various methods to manipulate both residential and commercial loads in real-time. As a result, various demand response (DR) frameworks are gaining popularity because of their ability to shape electricity demand by broadcasting time-varying signals to customers; however, most aggregators have not implemented complex DR programs beyond peak shaving and emergency load reduction initiatives. One reason for this is the customers' unknown and time-varying responses to dispatch signals, which can lead to economic uncertainty for the aggregator and reliability concerns for the grid \cite{RTP_DR}. The aggregator could explicitly request price sensitivity information from its customers; however, this two-way negotiation has a large communication overhead and most customers cannot readily characterize their price sensitivities, and even if they could, they might not be willing to share this private information. \textcolor{black}{As such, aggregators prefer the 1-way \textit{passive} approach because it does not require any real-time feedback from the customer and it does not require new communication infrastructure  for reporting preferences (e.g., a web portal, phone application, etc.).} With this in mind, future load shaping initiatives for renewable integration (i.e., more complex objectives than peak shaving) need to be able to \textit{passively} learn customers' response to dispatch signals only from historical data of past interactions \cite{gomez2012learning}.

Recently, much work has been done for aggregators attempting to learn customers' price responses whilst implementing peak shaving DR programs. The authors of \cite{driz11} present a data-driven strategy to estimate customers' demands and develop prices for DR. In \cite{driz12}, the authors use linear regression models to derive estimations of customers' responses to DR signals. Similarly, \cite{driz13} develops a joint online learning and pricing algorithm based on linear regression. In \cite{driz15}, the authors present a contract-based DR strategy to learn customer behavior while broadcasting DR signals. The authors of \cite{driz17} present an online learning approach based on piecewise linear stochastic approximation for an aggregator to sequentially adjust its DR prices based on the behavior of the customers in the past. In \cite{dvorkin_9}, the authors develop a risk-averse learning approach for aggregators operating DR programs. In \cite{dvorkin_10}, a learning algorithm for customers' utility functions is developed and it is assumed that the aggregator acts within a two-stage (day-ahead and real-time) electricity market. \textcolor{black}{Additionally, the authors of \cite{zhou2016residential} present a learning framework for forecasting individual loads and DR capabilities and find that users with more variable consumption patterns are more effective DR participants.} Using a similar framework as in this work, a multi-armed bandit (MAB) formulation is used in  \cite{driz14,driz19} to determine which customers to target in DR programs.

In addition to learning how customers respond to DR signals, an  aggregator must also consider power system constraints to ensure reliable operation (e.g., nodal voltage, transformer capacities, and line flow limits). In real distribution systems, it is critical that these constraints are satisfied at every time step to ensure customers receive adequate service and to avoid potential grid failures  even without sufficient knowledge about how customers respond to price signals (i.e., in early learning stages) \cite{dall2017chance, mieth2018data}. \textcolor{black}{One paper similar to ours that considers these realistic constraints, \cite{mieth2018online}, presents a least-square estimator approach to learn customer sensitivities and implements DR in a distribution network. The authors of \cite{mieth2018online} show that their least-square algorithm’s parameter estimation error converges to zero over time, thus the algorithm’s regret is sublinear while also accounting for the distribution network’s constraints.}

Similar to the aforementioned papers, the work presented in this manuscript considers the problem of an aggregator \textit{passively} learning the customers' price sensitivities while running a load shaping program. However, our approach permits more complex load shaping objectives (e.g., tracking a daily target load profile) and varies in terms of both load modeling and learning approach from all the above papers. Specifically, we present a  multi-armed bandit (MAB) heuristic akin to Thompson sampling (TS) to tackle the trade-off between exploration of untested price signals and exploitation of well-performing price signals while ensuring grid reliability. It is important to note that the standard TS heuristic cannot guarantee that grid reliability constraints are upheld during the learning process. As such, we present two {\it modified} versions of TS while retaining the fundamental principles TS is based on. Furthermore, we provide discussion on how the constraints are upheld (i.e., operational reliability) for the modified heuristics, discussion on the performance of the  heuristics compared to a clairvoyant solution, and simulation results highlighting the strengths of the method. 

\textcolor{black}{In our work, we make use of a load clustering technique in order to {\it exploit the known physical structure} of the problem and make use of our prior knowledge of how flexible electric appliances behave to lower the problem dimensionality. We note that grouping (clustering) loads for dimensional reduction is common in DR literature \cite{7085625}. Some pertinent examples include \cite{mathieu2012state} where the authors aggregate heterogeneous thermostatically controlled loads (TCLs) using an LTI “bin” model, \cite{foster2013optimal} where the authors group EVs into “classes” depending on their charging availability, \cite{8107585} where the authors present a load profile clustering method for load data classification based on information entropy, piecewise aggregate approximation, and spectral clustering, \cite{6545388} where the authors present aggregate models for classes of TCLs that include statistical information of the population, systematically deal with heterogeneity, and account for a second-order effects, \cite{7364011} where the authors propose a clustering technique for determining natural segmentation of customers and identification of temporal consumption patterns in the smart grid domain, and \cite{7206601} where the authors develop cohorts, or groups of consumers with similar consumption patterns, from correlations between daily loads.}

\textcolor{black}{The \textbf{main contributions} of this work are as follows:
\begin{itemize}
    \item We use the multi-armed bandit (MAB) framework to model the stochastic and unknown nature of customers' daily aggregate response to electricity prices.
    \item We make use of an appliance clustering methodology to provide a mesoscopic model of the price responsive demand of a large population of flexible  appliances and reduce the dimensionality of the learning problem.
    \item Our learning framework can account for daily variabilities and realistic grid reliability constraints that are critical for daily operation in spite of uncertainty about customers' price response.
    \item We present two modified heuristics based on Thompson sampling (TS) as solutions to the constrained learning and pricing problem.
    \item We provide a performance guarantee in the form of a regret bound and discussion on the reliability guarantees of the approach as well as a distribution system case study demonstrating the efficacy of the approach.
\end{itemize}
}

The remainder of the paper is organized as follows: Section \ref{section: Problem Formulation} presents the aggregator's daily objective as well as the customers' load model. Section \ref{section: MAB} describes the multi-armed bandit formulation for the electricity pricing problem, presents the modified TS heuristic, and discusses its performance and reliability. Section \ref{section: simulation} presents simulation results that showcase the efficacy of the approach. The Appendix contains a table of notation and proofs.

\section{Problem Formulation}
\label{section: Problem Formulation}

\subsection{The Aggregator's Objective}
\label{subsection: DSO's Objective}
The aggregator's main goal is to select dispatch signals to manipulate customer demand according to a given optimization objective that varies daily.  Specifically, we consider the case where the aggregator broadcasts a dispatch signal $\mathbf{p}_{\tau} = [p(t)]_{t=1,\dots,T}$ to the population of customers each day (we use $t=1,\dots,T$ to index time of day and $\tau=1,\dots,\mathcal{T}$ to index days). The set of dispatch signals available for use by the aggregator is denoted as $\mathcal P$.
In this paper, without a loss of generality, we will assume that the dispatch signal sent to customers for load shaping purposes is a real-time pricing (RTP) signal\footnote{The reader should \textcolor{black}{note} that this choice is not fundamental to the development of the modified learning heuristics we present in this paper. It only allows us to provide a concrete characterization of the response to dispatch signals by mathematically modeling the customers as cost-minimizing  agents equipped with home energy management systems in Section \ref{subsection: Cluster Price Response}.}. 

\textcolor{black}{The aggregator's cost function could cover a broad range of goals including (but not limited to) manipulating the population’s load to match a target profile, minimizing the distribution grid’s electricity cost from the regional retailer, or solving for the dispatch of multiple generators, if a market is operated at the distribution system level.}

\textcolor{black}{In this work,} on each day $\tau$, \textcolor{black}{we assume} the aggregator's cost function is a fixed and known nonlinear function $f(\mathbf{D}_{\tau}^{}(\mathbf{p}_{\tau}), \mathbf{V}_{\tau})$ that depends on the load profile $\mathbf{D}_{\tau}(\mathbf{p}_{\tau})$ of the population in response to the daily broadcasted price $\mathbf{p}_{\tau}$ and a random exogenous parameter vector $\mathbf{V}_{\tau}$\footnote{\textcolor{black}{We note that the function $f$ need not have a closed form representation and thus can represent the solution of an economic dispatch problem with multiple generators, which can still be handled through our framework. However, without loss of generality and purely for brevity of notation, here we focus on common distribution systems which usually lack two-sided markets, and thus we focus on load profile manipulation for renewable integration  and distribution system protection.}}. \textcolor{black}{The population's load profile on day $\tau$, $\mathbf{D}_{\tau}(\mathbf{p}_{\tau})$, is a $T\times1$ vector with the $t^{th}$ element corresponding to the population's power demand during time period $t$.} The exogenous and given \textcolor{black}{$T\times1$} vector $\mathbf{V}_{\tau}$ varies daily and can correspond to a daily target profile reflecting renewable generation forecasts, weather predictions, and grid conditions. We consider the exogenous vectors to be i.i.d. drawn from a distribution defined on a finite sample space $\mathcal{V}$, with each outcome  drawn with a nonzero probability. \textcolor{black}{We would like the reader to note that this assumption is only made for convenience for our theoretical regret performance guarantee in  Theorem \ref{thm conTS}. In a real-world implementation, the daily exogenous parameters could be correlated across days (e.g., due to weather, seasons, weekday/weekend, etc.). However, this correlation does not affect the safety guarantees of our algorithm or its applicability (i.e., it only affects our formal regret results).} 

The aggregator must ensure that the broadcasted price signals do not result in load profiles that violate distribution system reliability constraints (e.g., nodal voltage, transformer capacities, or line flow limits). As such, if the aggregator had full information about how the population responds to price signals (i.e., full knowledge of $\mathbf{D}_{\tau}^{}(\mathbf{p}_{\tau})$), the aggregator can solve the following optimization problem on day $\tau$ to select the optimal price $\mathbf{p}^{\star}_{\tau}$:
\begin{align}\label{clairvoyant}
     &\mathbf{p}_{\tau}^{\star} =
        \argmin_{\mathbf{p}_{\tau}\in\mathcal{P}}  f\big(\mathbf{D}_{\tau}^{}(\mathbf{p}_{\tau}), 
        \mathbf{V}_{\tau}\big)\\
        &\text{\hspace{24pt}s.t.\hspace{24pt}} g_j\big(\mathbf{D}_{\tau}^{}(\mathbf{p}_{\tau})\big)\leq0,\;\;\forall j=1,\dots,J
\end{align}
where $g_j(\cdot)_{j=1,\dots,J}$ is used to represent the reliability constraints for the distribution system. \textcolor{black}{We note that these general constraints need not be linear for the proposed Thompson sampling approach.}

However, as explained in the introduction, knowledge of customers' price response is unavailable to the aggregator. Recall, 1) the aggregator does not want to directly query customers for their price sensitivities, 2) most customers cannot readily characterize their price sensitivities, and 3) customers  might  not  be  willing  to  share this private information. Accordingly, the aggregator needs a method to sequentially choose daily price signals to simultaneously 1) control their daily incurred cost; 2) learn the customers' price response models; and 3) ensure the distribution system constraints are not violated at any time. 

\subsection{Distribution System Operational Constraints}
\label{subsection: Network Constraints}
As stated previously, there are various operational constraints within a distribution system that should be met in order to ensure adequate service for customers and to prevent grid failures. In the aggregator's daily optimization in Section \ref{subsection: DSO's Objective}, the constraints are formulated as general functions $g_j(\cdot)_{j=1,\dots,J}$. Specifically, these general functions represent  distribution system parameters  (i.e., the nodal voltage $u_{\tau}(t)$ and power flow through distribution lines $f_{\tau}(t)$) that should obey the following constraints:
\begin{align}
    \label{eqn: voltage min const}
    u_{\tau}(t) &\geq u^{min}, \;\;\forall t,\tau,\\
    \label{eqn: voltage max const}
    u_{\tau}(t) &\leq u^{max}, \;\;\forall t,\tau,\\
    \label{eqn: power flow const}
    f_{\tau}(t) &\leq S^{max}, \;\;\forall t,\tau,
\end{align}
where $u^{min}$, $u^{max}$, and $S^{max}$ correspond to the lower voltage limit, upper voltage limit, and power flow limit, respectively, for the population's connection to the distribution grid. We note that $u_{\tau}(t)$ and $f_{\tau}(t)$ can be easily derived from the population's load profile $\mathbf{D}_{\tau}^{}(\mathbf{p}_{\tau})$ (See Section \ref{subsection: Power Flow Model}). Now that we have described the aggregator's objective and the distribution system's constraints, we next describe the customers' load model as well as their price response model.

\subsection{Load Flexibility Model}
\label{subsection: Electricity Consumption Clusters}

It is hard to approach the problem of learning the response of a population of customers to complex dispatch signals such as RTP as a complete ``black box problem'', i.e., by just observing the broadcasted price and the load response. There are many reasons for this, including 1) the existence of random or exogenous parameters  which lead to variability in the  temporal and geographical behavior of electricity demand; 2) the variability of the control objective on a daily basis (e.g., due to randomness in renewable generation outputs, market conditions, or baseload); and 3) the small size of the set of observations that one can gather compared to the high dimensional structure of the load (there are only 365 days in a year, so only 365 \textcolor{black}{sets} of prices can be posted). Hence, in this paper, we will be {\it exploiting the known physical structure} of the problem and making use of our statistical   prior knowledge of how the load behaves to lower the problem dimensionality.

Specifically, to lower the dimensionality for the learning problem, we explore the fact that flexible loads only show limited number of ``load signatures'' \textcolor{black}{(justified due to the automated nature of load response through home energy management systems, the limited types of flexible appliances, and the common electricity usage patterns that emerge from electricity customers as shown in \cite{zhou2016residential,kwac2014household})}. Let us assume that electric appliances can belong to a finite number of clusters $c\in\mathcal{C}$. For each cluster $c$, we denote $\mathcal{D}_c^{}$ as the set of feasible daily power consumption schedules that satisfy the energy requirements of the corresponding appliances. Any power consumption schedule, $[d_{c}(t)]_{t=1,\dots,T} = \mathbf{D}_c \in \mathcal{D}_c$, would satisfy the daily power needs of an appliance in cluster $c$. For example, consider a cluster that represents plug-in electric vehicles (EVs) that require $E_c$ kWh in the time interval $[t_1, t_2]$ with a maximum charging rate of $\rho_c$ kW. Accordingly, the set $\mathcal{D}_c$ of daily feasible power consumption schedules is given by:
\begin{align}
\label{eqn: cluster_example}
    \mathcal{D}_c = \Big\{ \mathbf{D}_c | \sum_{t=t_1}^{t_2} d_c(t) = E_c;\; 0\leq d_c(t) \leq \rho_c\Big\}.
\end{align}
\textcolor{black}{Another specific cluster example is that of electric appliances that are uninterruptible but can perform load shifting (e.g., a dishwasher cannot be interrupted once it is turned on but the start time of the cycle can be shifted). Let $\Pi_c(\cdot)$ denote the load profile of uninterruptible cluster $c$ appliances once they are turned on. For example, $\Pi_c(\cdot)$ could be a rectangular pulse function that outputs the rated power of the appliance, $\rho_c$ for the duration of the appliance's cycle and $0$ otherwise. To relay their load flexibility, cluster $c$ users can specify a time interval $[t_{c,1},t_{c,2}]$ within which the appliance cycle must start (e.g., a user wants the dishwasher to be finished before dinner). Thus the home energy management system can calculate the best values for the time shift, denoted by $t_c$, as long as it lies within the interval $[t_{c,1},t_{c,2}]$. The set $\mathcal{D}_c$ of daily feasible power consumption schedules for appliances in this cluster is given by:
\begin{align}
\label{eqn: cluster_example2}
    \mathcal{D}_c = \Big\{ \mathbf{D}_c | & d_c(t) =  \Pi_c(t-t_c);\;t_{c}\in [t_{c,1},t_{c,2}]\; \Big\}.
\end{align}}
\noindent\textcolor{black}{For discussion on characterizing the sets for other flexible appliances, including interruptible (Section III.B in \cite{mahnoosh_cluster_driz_20}, non-interruptible (Section III.D in \cite{mahnoosh_cluster_driz_20}), and thermostatically controlled loads (Section III.C in \cite{mahnoosh_cluster_driz_20}), we refer the reader to reference \cite{mahnoosh_cluster_driz_20}.}

By adopting this model, the total power consumption flexibility of a population of customers can be characterized as a function of how many appliances belong to each cluster within the population. Let us denote $a_{c}$ as the number of appliances in cluster $c$ (note that this will vary on a day by day basis as described in the next section). With this notation, we can write the set of feasible daily power consumption profiles for the population, $\mathcal{D}_{}$:
\begin{align}\label{minkowski}
    \mathcal{D}_{} = \sum_{c\in\mathcal{C}} a_{c} \mathcal{D}_c,
\end{align}
where the summation and scalar multiplication operations are defined in the sense of Minkowski addition\footnote{For two sets $A$ and $B$ defined on a finite dimensional Euclidean space, the Minkowski sum is defined as $A+B = \{\mathbf{a}+\mathbf{b}~|~\mathbf{a}\in A, \mathbf{b}\in B\}$.}.

\textcolor{black}{We would like to note that choosing the number of clusters in the model is a control knob that can be tuned by the aggregator as shown in \cite{alizadeh2011information}. Using a higher number of potential appliance clusters will increase the accuracy of the load model (i.e., reduce the quantization error in the reproduction of the individual load profiles) and yield better performance in the daily optimization once the true parameters have been sufficiently learned by the aggregator. However, increasing the number of load clusters increases the size of the problem space and increases the randomness in the customers’ daily loads thus slowing down the learning rate of the algorithm. The number of clusters will vary depending on the system being analyzed as well as the aggregator’s preferences.}

\subsection{Price Response Model}
\label{subsection: Cluster Price Response}
In this section, we discuss how the total population responds to dynamic electricity prices given the load flexibility model in \eqref{minkowski} \textcolor{black}{and how clustering is used to reduce the dimensionality of the problem}.  There are two main ways dynamic pricing affects the power consumption: 1) \textit{Automated per cluster response}: Within each load  cluster $c$ (i.e., given pre-specified preferences such as EV charging deadlines or AC temperature set points), we assume that the customer chooses the  power consumption profile $\mathbf{D}_c \in \mathcal{D}_c$ that minimizes their electricity cost dependent on the daily broadcasted price $\mathbf{p}_{\tau}$. For appliances in cluster $c$ on day $\tau$, we assume all will choose the same minimum cost  power consumption profile:
\begin{align}
\label{eqn: auto response}
    \widetilde{\mathbf{D}}_{c,\tau}(\mathbf{p}_{\tau}) = \argmin_{\mathbf{D}_c \in \mathcal{D}_c} \sum_{t=1}^T p(t) d_{c}(t).
\end{align}
\textcolor{black}{We assume that each appliance will always choose the cost minimizing power consumption profile out of the available profile set to combat the fact that the available profile sets $\mathcal{D}_c$ for each cluster can be infinitely large. Thus, we have effectively reduced the dimensionality of the problem as we know \textit{a priori} how each cluster will respond to each price signal (i.e., each cluster will always select its cost minimizing profile).} Due to the automated nature of home energy management systems, each cluster selecting its cost minimizing profile is a reasonable assumption once the customers have defined their flexibility preferences, e.g., the desired charge amounts and deadlines  for EVs \cite{chang2012coordinated, alizadeh2013least}. 2) \textit{Preference Adjustment}: We also consider the fact that customers may respond to price signals by adjusting their preferences. \textcolor{black}{Consider the following example: two customers (Customer-A and Customer-B) live in the same neighborhood but have different sensitivities to electricity prices. If electricity prices are high on a hot summer day, Customer-A might shutdown their air conditioner to avoid a large electricity bill; however, Customer-B prioritizes comfort over cost-savings, and leaves their air conditioner on, no matter the cost. As shown in the previous example, the number of appliances in each cluster, i.e., $a_{c}$ in \eqref{minkowski}, also depends on the daily posted price vector $\mathbf{p}_{\tau}$, and are now denoted as $a_{c}(\mathbf{p}_{\tau})$. }

Combining the  \textit{automated per cluster response} and \textit{preference adjustment}, we can define the population's load on day $\tau$ in response to the posted price $\mathbf{p}_{\tau}$ as follows:
\begin{align}
    \mathbf{D}_{\tau}^{\star}(\mathbf{p}_{\tau}) = \sum_{c\in\mathcal{C}} a_{c}(\mathbf{p}_{\tau}) \widetilde{\mathbf{D}}_{c,\tau}(\mathbf{p}_{\tau}).
\end{align}
As stated before, if the aggregator has full knowledge of the customers' price responses, which reduces to having full knowledge of the preference adjustments $a_{c}(\mathbf{p}_{\tau})$, then the aggregator can pick the daily price vector $\mathbf{p}_{\tau}^{\star}$ in order to shape the population's power consumption according to \eqref{clairvoyant}. However, as we cannot assume this, we model the  $a_{c}(\mathbf{p}_{\tau})$'s as random variables with parameterized distributions, $\phi_c$, based on the posted price signal $\mathbf{p}_{\tau}$ and an unknown but constant parameter vector $\boldsymbol{\theta}_{}^{\star}$. Here, $\boldsymbol{\theta}_{}^{\star}$ represents the \textit{true model} for the customers' sensitivity to the price signals. This allows for the complex response of the customer population to be represented as a single vector, thus reducing the dimensionality of the problem. \textcolor{black}{We note that while $a_{c}(\mathbf{p}_{\tau})$ may only take integer values in reality, we believe it is justified to relax this integrality constraint and allow it to take continuous values with large enough appliance population size.} With this in mind, we would like to highlight three properties of the price response model:
\begin{enumerate}
    \item The preference adjustment models $a_{c}(\mathbf{p}_{\tau})$ are stochastic and their distributions $\phi_c$ are parameterized by $\mathbf{p}_{\tau}$ and $\boldsymbol{\theta}_{}^{\star}$. This is due to exogenous factors outside of the aggregator's scope that influence customers' power consumption profiles resulting in a level of stochasticity in the responses to prices (i.e., customers will not respond to prices in the same fashion each day).
    
    \item The probability distributions of $a_{c}(\mathbf{p}_{\tau})$ (i.e., $\phi_c$) are unknown to the aggregator, i.e., the aggregator does not know the  true parameter $\boldsymbol{\theta}_{}^{\star}$ of the stochastic model. 
   
    \item The realizations of $a_{c}(\mathbf{p}_{\tau})$ are not directly observable by the aggregator. The aggregator can only monitor the population's total consumption profile $\mathbf{D}_{\tau}$ and cannot observe the decomposed response of each cluster $a_{c}(\mathbf{p}_{\tau})\widetilde{\mathbf{D}}_{c,\tau}(\mathbf{p}_{\tau})$ independently. 
\end{enumerate}

Because we have introduced stochasticity to customers' price response models, we appropriately alter the aggregator's optimization problem for selecting the price signal on day $\tau$ to account for the distributions $\phi_c$:
\begin{align}\label{NOTclairvoyant}
     &\mathbf{p}_{\tau}^{\star} =
        \argmin_{\mathbf{p}_{\tau}\in\mathcal{P}} \mathbb{E}_{\{\phi_c\}_{c\in\mathcal{C}}}\big[ f\big(\mathbf{D}_{\tau}^{}(\mathbf{p}_{\tau}), 
        \mathbf{V}_{\tau}\big)\big]\\
        \label{NOTclairvoyant_constraint}
        &\text{\hspace{0pt}s.t.\hspace{13pt}} \mathbb{P}_{\{\phi_c\}_{c\in\mathcal{C}}}\big[g_j\big(\mathbf{D}_{\tau}^{}(\mathbf{p}_{\tau})\big)\leq0\big]\geq 1-\mu,\;\;\forall j
\end{align}
where $\mu$ is the aggregator's desired reliability metric for the distribution system constraints. In \eqref{NOTclairvoyant}, the aggregator now considers minimizing an expected cost and is subject to probabilistic reliability constraints in \eqref{NOTclairvoyant_constraint} that depend on the distributions $\phi_c$ of the preference adjustment models $a_c(\mathbf{p}_\tau)$. 

\textcolor{black}{We note that the formulated chance constraints are enforced with respect to uncertainty in the clusters' price sensitivity parameters, not to the exogenous context vector $\mathbf{V}_{\tau}$. In this work, we assume the daily exogenous vector is fully known each day and does not add uncertainty to the problem. However, uncertainties in the exogenous vector are important to real-world systems such as the power grid and can be accommodated by our approach by adding external noise to these vectors in the same fashion as noise being added to the population's load. This, of course, would further slow down the learning rate of the algorithm due to the added noise reducing the effectiveness of each posterior update.}

Clearly, the aggregator needs to learn the underlying parameters of the stochastic models $\phi_c$ of how customers respond to price signals in order to select price signals for load shaping initiatives (i.e., the aggregator needs to learn $\boldsymbol{\theta}^{\star}$). Our proposed learning approach and pricing strategy for an electricity aggregator is detailed in the next section.

\section{Real-Time Pricing via Multi-Armed Bandit}
\label{section: MAB}

\subsection{Multi-Armed Bandit Overview}
\label{subsection: MAB formulation}

We utilize the multi-armed bandit (MAB) framework to model the iterative decision making procedure of an aggregator implementing a daily load shaping program \cite{krishnamurthy2018semiparametric, foster2018practical, xu2013thompson}. \textcolor{black}{The MAB problem can be described as a decision making problem where an agent has a set of available actions but can only take one action per round. After an action is taken, \textcolor{black}{the agent experiences a cost that is }dependent on the action taken. The agent can only learn about the distribution of costs from each action by experimenting. Throughout this iterative procedure, the agent faces the core dilemma: should the agent \textit{exploit} actions that have yielded small costs, or \textit{explore} actions that have not been tested thoroughly? The goal in a MAB problem is to develop a strategy for selecting actions that balance this trade-off and minimize the cumulative cost over a given time span. More thorough explanation and background of the MAB problem can be found in \cite{russo2018tutorial}.}

For the electricity pricing problem, the MAB framework exemplifies the \textit{exploration-exploitation} trade-off dilemma faced by an aggregator each day. Namely, should the aggregator choose to broadcast untested prices (i.e., \textit{explore}) to learn more information about the customers? Or should the aggregator choose to broadcast well-performing prices (i.e., \textit{exploit}) to manipulate the daily electricity demand? 

To evaluate the performance of an algorithm that aims to tackle the exploration-exploitation trade-off, one commonly examines the algorithm's \textit{regret}. Formally, regret is defined as the cumulative difference in cost incurred over $\mathcal{T}$ days between a clairvoyant algorithm (i.e., the optimal strategy that is aware of the customers' price responses) and any proposed algorithm that does not know the customers' price responses:
\begin{align}
\label{eqn:obj regret}
    R_{\mathcal{T}}=\sum_{\tau=1}^{\mathcal{T}} f(\mathbf{D}_{\tau}^{}(\mathbf{p}_{\tau}), 
        \mathbf{V}_{\tau})-f(\mathbf{D}_{\tau}^{}(\mathbf{p}^{\star}_{}), 
        \mathbf{V}_{\tau}).
\end{align}

Instead of considering the cumulative difference in objective function value, an alternative metric for regret is to count the number of times that suboptimal price signals are selected over the $\mathcal{T}$ days. For this, we introduce the following notation: let $\mathbf{p}^{\mathbf{V}_{\tau},\star}$ denote the optimal price signal for the true model of the population's price response $\boldsymbol{\theta}_{}^{\star}$ when the daily exogenous parameter $\mathbf{V}_{\tau}$ is observed on day $\tau$. Any price signal $\mathbf{p}_{\tau}\neq\mathbf{p}_{}^{\mathbf{V}_{\tau},\star}$ is considered a suboptimal price. Moreover, we denote $N_{\tau}(\mathbf{p},\mathbf{V})$ as the number of times up to day $\tau$ that the algorithm simultaneously observes the exogenous parameter $\mathbf{V}$ and selects the price signal $\mathbf{p}$. 
As such, the total number of times that suboptimal price signals are selected over $\mathcal{T}$ days is:
\begin{align}
\label{eq: alt regret}
    \sum_{\mathbf{V}\in\mathcal{V}}
    \sum_{\mathbf{p}\in\{\mathcal{P}\setminus\mathbf{p}^{\mathbf{V},\star}\}}
    N_{\mathcal{T}}(\mathbf{p},\mathbf{V})
    =
    \sum_{\tau=1}^{\mathcal{T}}
    \mathbbm{1}^{\{\mathbf{p}_{\tau}\neq\mathbf{p}_{}^{\mathbf{V}_{\tau},\star}\}},
\end{align}
where $\mathbbm{1}^{\{\cdot\}}$ is the indicator function that is set equal to one if the criteria is met and zero otherwise. Subsequently, in an iterative decision making problem such as this, the question arises: \textit{how can an aggregator learn to price electricity with bounded regret, and what are the regret bounds we can provide for a proposed algorithm given dynamically changing grid conditions and reliability constraints?} In the following sections, we present a modified Thompson sampling heuristic for the electricity pricing problem to simultaneously learn the true model $\boldsymbol{\theta}_{}^{\star}$ for the population, select the daily price signals, ensure grid reliability, and provide a  regret guarantee.


\subsection{Thompson Sampling}
\label{subsection: TS}

Thompson sampling (TS) is a well-known MAB heuristic for choosing actions in an iterative decision making problem with the \textit{exploration-exploitation} dilemma \cite{russo2018tutorial, russo2014learning, agrawal2012analysis}. \textcolor{black}{Two other well-studied frameworks, greedy algorithms and upper-confidence bound (UCB) algorithms, have shown promise in this problem area. However, greedy algorithms are inferior to Thompson sampling in regret performance and UCB algorithms are restricted to simpler linear optimizations \cite{auer2002using,dani2008stochastic,abbasi2011improved}, whereas Thompson sampling can readily handle more general objective functions such as those adopted in our paper \cite{gopalan2014thompson}. Additionally, a novel aspect of our paper is that we have shown how to modify the Thompson sampling heuristic to account for reliability constraints with a theoretical guarantee (Proposition 1). There are no other bandit optimization approaches known  to be able to handle general objective functions with safety constraints. Relevant works here include the analysis of the performance of the UCB algorithm in the linear MAB setting with linear safety-constraints \cite{amani2019linear}, and well as linear TS with linear constraints  \cite{moradipari2019safe}.  In the latter work, it is shown that in the linear case, the presence of linear constraints do not negatively affect the regret performance of TS, which is remarkable and could be a preliminary justification as to why TS performs well in our paper in the presence of general (non-linear) cost and constraint functions.}

Simply put, the integral characteristic of Thompson sampling is that the algorithm's knowledge on day $\tau$ of the unknown parameter $\boldsymbol{\theta}_{}^{\star}$ is represented by the prior distribution $\pi_{\tau-1}$. Each day the algorithm samples $\boldsymbol{\widetilde{\theta}}_{\tau}$ from the prior distribution, and selects an action assuming that the sampled parameter is the true parameter. The algorithm then makes an observation dependent on the chosen action and the hidden parameter and performs a Bayesian update on the parameter's distribution $\pi_{\tau}$ based on the new observation. Because TS samples parameters from the prior distribution, the algorithm has a chance to explore (i.e., draw new parameters) and can exploit (i.e., draw parameters that are likely to be the true parameter) through out the run of the algorithm.

\subsection{Constrained Thompson Sampling}
\label{subsection: Con-TS-RTP}

In this section, we present the MAB heuristic titled Con-TS-RTP adopted to the electricity pricing problem. Con-TS-RTP is a modified Thompson sampling algorithm where the daily optimization problem is subject to constraints (standard TS algorithms do not have constraints)\cite{conTS2019}.

\textcolor{black}{When initializing $\pi_0$, the initial distribution on the customers’ unknown parameter can be selected by the aggregator. If the aggregator has access to prior information regarding the true parameter, then they could initialize the prior as a distribution of their choice. However, if the aggregator has no prior knowledge, a uniform distribution among all available parameters may be used to model the lack of knowledge of the aggregator.}

Each day, the algorithm observes the daily target profile $\mathbf{V}_{\tau}$, draws a parameter $\boldsymbol{\widetilde{\theta}}_{\tau}$ from the prior distribution, broadcasts a price signal to the customers, observes the load profile of the population in response to the broadcasted price, and then performs a Bayesian update on the parameter's distribution $\pi_{\tau}$ based on the new observation. \textcolor{black}{We note that there are no restrictions on the class of optimization problem to be solved each day; however, in order for our regret guarantee to hold, the aggregator\textit{ must }be able to find the globally optimal solution and can use any desired solution method to do so. In our experimental examples, we assume that  $\theta$'s and $\mathbf{p}_{\tau}$'s are chosen from discrete sets in order to be able to guarantee that an enumeration method could solve for the globally optimal price signals each day in spite of non-convexities that arise. }

The observation on day $\tau$ is denoted as $\mathbf{Y}_{\tau} = \mathbf{D}_{\tau}^{\star}(\mathbf{p}_{\tau})$ and we assume that each $\mathbf{Y}_{\tau}$ comes from the observation space $\mathcal{Y}$ that is known a priori. When performing the Bayesian update, the algorithm makes use of the following \textit{likelihood function}: $\ell(\mathbf{Y}_{\tau};\mathbf{p},\boldsymbol{\theta})=\mathbb{P}_{\theta}(\mathbf{D}_{\tau}^{\star}(\mathbf{p}_{\tau}) = \mathbf{Y}_{\tau}|\mathbf{p}_{\tau}=\mathbf{p})$. This function calculates the likelihood of observing a specific load profile when broadcasting price $\mathbf{p}$ and the true parameter is $\boldsymbol{\theta}$. The pseudocode for Con-TS-RTP applied to the constrained electricity pricing problem is presented in Algorithm \ref{alg: conTS}.
\begin{algorithm}[]
    \caption{\textsc{Con-TS-RTP}}
    \label{alg: conTS}
    \begin{algorithmic}
    \STATE \textbf{Input:} Parameter set $\Theta$; Price set $\mathcal{P}$; Observation set $\mathcal{Y}$; 
    Voltage constraints $u_{}^{min},u_{}^{max}$; Power flow constraint $S_{}^{max}$, Reliability metrics $\mu$, $\nu$
    \STATE \label{alg.init.pie}\textcolor{black}{\textbf{Initialize} $\pi_{0}$ based on aggregator's available prior knowledge of customer sensitivity.}
    \end{algorithmic}
    \begin{algorithmic}[1]
    \FOR{Day index $\tau=1...\mathcal{T}$}
    \STATE \textcolor{black}{Sample the daily hidden parameter $\boldsymbol{\widetilde{\theta}}_{\tau}\in\Theta$ from the aggregator's prior distribution $\pi_{\tau-1}$.}
    \STATE Observe the daily exogenous parameter $\mathbf{V}_{\tau}$.\\
    \STATE Broadcast the daily price signal:\\
    \vspace{-10pt}
    \begin{align*}
        &\hspace{-0pt}\hat{\mathbf{p}}_{\tau} =
        \argmin_{\mathcal{P}} \mathbb{E}_{\{\phi_c\}_{c\in\mathcal{C}}}\big[ f(\mathbf{D}_{\tau}^{}(\mathbf{p}_{\tau}), 
        \mathbf{V}_{\tau})|\boldsymbol{\theta}=\boldsymbol{\widetilde{\theta}}_{\tau}\big]\\
        &\text{\hspace{-0pt}Subject to:}\\
        &\text{\hspace{-20pt}\textbf{\textit{Constraint Set A}}:}\\
        &\hspace{-20pt}\begin{cases}
            \textit{\textbf{A.1: }} \mathbb{P}_{\{\phi_c\}_{c\in\mathcal{C}}}[u_{\tau}(t) \geq u_{}^{min} |\boldsymbol{\theta}=\boldsymbol{\widetilde{\theta}}_{\tau}
            ]\geq 1-\mu, \;\;\forall t\\
            \textit{\textbf{A.2: }}
            \mathbb{P}_{\{\phi_c\}_{c\in\mathcal{C}}}[u_{\tau}(t) \leq u_{}^{max}
            |\boldsymbol{\theta}=\boldsymbol{\widetilde{\theta}}_{\tau}
            ]\geq 1-\mu, \;\;\forall t\\
            \textit{\textbf{A.3: }}
            \mathbb{P}_{\{\phi_c\}_{c\in\mathcal{C}}}[f_{\tau}(t) \leq S_{}^{max}
            |\boldsymbol{\theta}=\boldsymbol{\widetilde{\theta}}_{\tau}
            ]\geq 1-\mu, \;\;\forall t\\
        \end{cases}\\
        &\text{\hspace{-20pt}\textbf{\textit{Constraint Set B}}:}\\
        &\hspace{-20pt}\begin{cases}
            \textit{\textbf{B.1: }} \mathbb{P}_{\{\phi_c\}_{c\in\mathcal{C}}}[u_{\tau}(t) \geq u_{}^{min} |\boldsymbol{\theta}\sim\pi_{\tau-1}
            ]\geq 1-\nu, \;\;\forall t\\
            \textit{\textbf{B.2: }}
            \mathbb{P}_{\{\phi_c\}_{c\in\mathcal{C}}}[u_{\tau}(t) \leq u_{}^{max}
            |\boldsymbol{\theta}\sim\pi_{\tau-1}
            ]\geq 1-\nu, \;\;\forall t\\
            \textit{\textbf{B.3: }}
            \mathbb{P}_{\{\phi_c\}_{c\in\mathcal{C}}}[f_{\tau}(t) \leq S_{}^{max}
            |\boldsymbol{\theta}\sim\pi_{\tau-1}
            ]\geq 1-\nu, \;\;\forall t\\
        \end{cases}
    \end{align*}
    \vspace{2pt}
    \STATE \textcolor{black}{Observe the population's load response to price $\mathbf{p}_{\tau}$: $\mathbf{Y}_{\tau} = \mathbf{D}_{\tau}^{\star}(\mathbf{p}_{\tau})$.}
    \STATE \textcolor{black}{Update the aggregator's knowledge of the true parameter in the posterior:}
    \begin{align*}
        &\forall S\subseteq \Theta: \pi_{\tau}(S) = \frac{\int_S \ell(\mathbf{Y}_{\tau};\hat{\mathbf{p}}_{\tau},\boldsymbol{\theta})\pi_{\tau-1}(d\boldsymbol{\theta})}{\int_{\Theta} \ell(\mathbf{Y}_{\tau};\hat{\mathbf{p}}_{\tau},\boldsymbol{\theta})\pi_{\tau-1}(d\boldsymbol{\theta})}
    \end{align*}
    \ENDFOR
    \end{algorithmic}
\end{algorithm}

\subsection{Discussion on Regret Performance of Con-TS-RTP}
\label{subsection: conTS Regret}

The regret analysis of Con-TS-RTP is inspired by the results in  \cite{gopalan2014thompson} for TS with nonlinear cost functions. The authors in \cite{drizzy} extended the regret results from \cite{gopalan2014thompson} by analyzing the effects of an objective function that is dependent on exogenous parameters such as $\mathbf{V}_{\tau}$. The analysis in the aforementioned papers provides bounds on the total number of times that suboptimal price signals selected by the algorithm over $\mathcal{T}$ days as specified in equation \eqref{eq: alt regret}. 
The regret guarantee we provide in this work extends the result further, allowing for constraints in the daily optimization that are dependent on the sampled $\boldsymbol{\widetilde{\theta}}_{\tau}$. As such, our regret guarantee applies to the Con-TS-RTP algorithm with constraints as formulated in \textit{Constraint Set A} in Algorithm \ref{alg: conTS}. We refer the reader to the appendix as well as \cite{drizzy} and \cite{gopalan2014thompson} for further discussion on the derivation of Theorem \ref{thm conTS}.
\begin{assumption}
\label{assump finite}
(Finitely many price signals, observations). $|\mathcal{P}|,|\mathcal{Y}|<\infty$.
\end{assumption}
\begin{assumption}
\label{assump grain of truth}
(Finite Prior,``Grain of truth") The prior distribution $\pi$ is supported over finitely many particles: $|\Theta|<\infty$. The true parameter exists within the parameter space: $\boldsymbol{\theta}_{}^{\star}\in\Theta$. The initial distribution $\pi_{0}$ has non-zero mass on the true parameter $\boldsymbol{\theta}^{\star}$ (i.e., $\mathbb{P}_{\pi_{0}}[\boldsymbol{\theta}^{\star}]>0$).
\end{assumption}
\begin{assumption}
\label{assump unique best action}
(Unique optimal price signal). There is a unique optimal price signal $\mathbf{p}^{\mathbf{V},\star}$ for each exogenous parameter $\mathbf{V}\in\mathcal{V}$.
\end{assumption}
\begin{theorem}
\label{thm conTS}
Under assumptions \ref{assump finite}-\ref{assump unique best action} and Constraint Set A in Algorithm \ref{alg: conTS}, for $\delta, \epsilon \in (0,1)$, there exists $\mathcal{T}^{\star}\geq0$ s.t. for all $\mathcal{T}\geq\mathcal{T}^{\star}$, with probability $1-\delta$:\\
\begin{align}
    \sum_{\mathbf{V}\in\mathcal{V}}
    \sum_{\mathbf{p}\in\{\mathcal{P}\setminus\mathbf{p}^{\mathbf{V},\star}\}}
    N_{\mathcal{T}}(\mathbf{p},\mathbf{V})\leq B + C(\log\mathcal{T}),
\end{align}
\vspace{1ex}

\noindent where $B\equiv B(\delta,\epsilon,\mathcal{P},\mathcal{Y},\Theta)$ is a problem-dependent constant that does not depend on $\mathcal{T}$, and $C(\log\mathcal{T})$ depends on $\mathcal{T}$, the sequence of selected price signals, and the Kullback-Leibler divergence properties of the bandit problem (i.e., the marginal Kullback-Leibler divergences of the observation distributions $\text{KL}\big[\ell(\mathbf{Y};\mathbf{p},\boldsymbol{\theta}^{\star}), \ell(\mathbf{Y};\mathbf{p},\boldsymbol{\theta})\big]$ (The complete description of the $C(\log\mathcal{T})$ term is left to the appendix).
\end{theorem}
\noindent\textit{Proof.} The proof is in the appendix.



In the next section, we discuss the distribution system reliability issue that arises from how the Con-TS-RTP algorithm handles the distribution system constraints (i.e., \textit{Constraint Set A}) and a modification to the Con-TS-RTP algorithm to ensure the constraints are enforced (i.e., \textit{Constraint Set B}).

\subsection{Con-TS-RTP with Improved Reliability Constraints}
\label{subsection: Thompson Sampling 2}

In order for the aggregator to ensure safe operation of the distribution grid while running the Con-TS-RTP algorithm, the reliability constraints need to hold for the true price response model $\boldsymbol{\theta}^{\star}$ each day. However, with the constraints formulated as in Algorithm \ref{alg: conTS}'s \textit{Constraint Set A}, the distribution system constraints are only enforced for the sampled $\boldsymbol{\widetilde{\theta}}_{\tau}$ and not necessarily the true parameter $\boldsymbol{\theta}^{\star}$. This entails that the distributions $\{\phi_c\}_{c\in\mathcal{C}}$ are parameterized by the sampled $\boldsymbol{\widetilde{\theta}}_{\tau}$; therefore, they are inaccurate if any parameter $\boldsymbol{\widetilde{\theta}}_{\tau}\neq\boldsymbol{\theta}^{\star}$ is sampled. This could potentially lead to many constraint violations throughout the run of the algorithm resulting in inadequate service for the customers and grid failures.

Due to the importance of reliable operation of the distribution system, we present a modification to the Con-TS-RTP algorithm (i.e., replacing \textit{Constraint Set A} with \textit{Constraint Set B} in Algorithm \ref{alg: conTS}) to increase the reliability of the selected prices and resulting load profiles with respect to the grid constraints. 
Specifically, we propose alternate constraints that depend on the algorithm's current knowledge of the true parameter, instead of the sampled parameter. In other words, instead of depending on $\boldsymbol{\widetilde{\theta}}_{\tau}$, the proposed alternate constraints depend on the prior distributions $\pi_{\tau-1}$ as follows:
\begin{align}
    &\label{prob1}\mathbb{P}_{\{\phi_c\}_{c\in\mathcal{C}}}[u_{\tau}(t) \geq u_{}^{min} |\boldsymbol{\theta}\sim\pi_{\tau-1}
            ]\geq 1-\nu, \;\;\forall t\\
    &\label{prob2}\mathbb{P}_{\{\phi_c\}_{c\in\mathcal{C}}}[u_{\tau}(t) \leq u_{}^{max}
            |\boldsymbol{\theta}\sim\pi_{\tau-1}
            ]\geq 1-\nu, \;\;\forall t\\
    &\label{prob3}\mathbb{P}_{\{\phi_c\}_{c\in\mathcal{C}}}[f_{\tau}(t) \leq S_{}^{max}
            |\boldsymbol{\theta}\sim\pi_{\tau-1}
            ]\geq 1-\nu, \;\;\forall t
\end{align}
where $\nu$ is a small constant (detailed in Proposition \ref{prop}). 
When considering constraints \eqref{prob1}-\eqref{prob3} in the Con-TS-RTP algorithm, the algorithm will select more conservative price signals each day that can guarantee the distribution system's constraints are met with high probability by using the information in the updated prior distributions. Before analyzing the modified algorithm's reliability, we make the following assumption:
\begin{assumption}
\label{assump 2}
There exists $\xi^{\star} > 0$ and $\lambda\geq0$, such that for all $\boldsymbol{\theta} \neq \boldsymbol{\theta}_{}^{\star}, \text{ KL}\big[\ell(\mathbf{Y};\mathbf{p},\boldsymbol{\theta}^{\star}), \ell(\mathbf{Y};\mathbf{p},\boldsymbol{\theta})\big] \geq \xi^{\star}$, where
\begin{align*}
    \nonumber\xi^{\star}_{\boldsymbol{\;\theta},\mathbf{p}}=\max_{x\in\mathbb{Z}^{>0}}\Bigg\{
    \frac{-\lambda}{x}-\frac{4}{\sqrt{x}}\sqrt{\frac{\log{|\mathcal{Y}||\mathcal{P}|}}{\delta}+\frac{\log{x}}{2}}
    \\\times\sum_{\mathbf{Y}\in\mathcal{Y}}\Big|\log{\frac{\ell(\mathbf{Y};\mathbf{p},\boldsymbol{\theta}^{\star})}{\ell(\mathbf{Y};\mathbf{p},\boldsymbol{\theta})}}\Big|\Bigg\}
\end{align*}
and
\begin{align*}
    \xi^{\star}=\max_{\boldsymbol{\theta}\in\Theta,\mathbf{p}\in\mathcal{P}}\xi^{\star}_{\boldsymbol{\;\theta},\mathbf{p}}.
\end{align*}
\end{assumption}
\noindent
Assumption \ref{assump 2} ensures that as the aggregator performs the steps in Algorithm 1, the algorithm's Bayesian updates of the prior distribution $\pi_{\tau}$ will likely never decrease the mass of the true parameter $\boldsymbol{\theta}^{\star}$ below a certain threshold. Specifically, with Assumption \ref{assump 2}, it can be shown (as in \cite{gopalan2014thompson}) that with probability $1-\delta\sqrt{2}$ the following holds for all $\tau\geq1$:
\begin{align}
\label{eqn: minimum mass}
    \pi_{\tau}(\boldsymbol{\theta}^{\star})
    \geq
    \pi_{0}(\boldsymbol{\theta}^{\star})
    e^{-\lambda|\mathcal{P}|},
\end{align}
where $\lambda\geq0$ is a chosen parameter (from Assumption \ref{assump 2}) that dictates the minimum reachable mass of the true parameter via Bayesian updating. With the modified constraints \eqref{prob1}-\eqref{prob3} and the minimum mass of the true parameter specified in \eqref{eqn: minimum mass}, the reliability of Con-TS-RTP can be characterized as follows:
\begin{proposition}
\label{prop}
Under assumptions \ref{assump finite}-\ref{assump 2}, with $\nu$ in equations \eqref{prob1}-\eqref{prob3} chosen such that $\nu\leq\mu\pi_{0}(\boldsymbol{\theta}^{\star})
e^{-\lambda|\mathcal{P}|}$,  with probability $1-\delta\sqrt{2}$, the Con-TS-RTP algorithm with Constraint Set B will uphold the probabilistic distribution system constraints as formulated in \eqref{NOTclairvoyant_constraint} for each day $\tau=1,\dots,\mathcal{T}$. 
\end{proposition}
\noindent\textit{Proof.} The proof is in the appendix.\\ 
\textcolor{black}{\noindent\textit{Remark:} The novelty of Con-TS-RTP is that we can ensure with high probability an unsafe price signal is never selected. We can tune the safety parameter to determine what level of risk is acceptable to the aggregator. We note that the selection of an unsafe price signal has no effect on the learning capability of the algorithm. The Con-TS-RTP algorithm will learn regardless of safe/unsafe price signals. The algorithm will never crash/stop prematurely due to the selection of an unsafe price signal; however, the local distribution grid might surpass safety limits on transformers or line flow limits due to an unsafe price selection on select very limited days, at which points protective measures (e.g., relays)  should be used to ensure physical grid safety. We note that this is natural for any learning algorithm dealing with stochastic conditions and unknown system parameters. Contingencies can never be avoided 100\%, similar to other grid operation paradigms that deal with uncertain conditions (e.g., wholesale market dispatch with renewables or possible transmission system contingencies). They could only be avoided with a certain high probability when making dispatch decisions. However, it is understood that other protective measures should always be put in place to avoid physical system damage in case of contingencies.}

\section{Experimental Evaluation}

\label{section: simulation}
\subsection{Test Setup: Radial Distribution System}
\label{subsection: Radial Distribution System}


In this section we describe the power distribution system and the corresponding network parameters for the test case. We consider an actual radial distribution system from the ComEd service territory in Illinois, USA (adopted from \cite{andrianesis2019locational} and shown in Fig. \ref{fig: radial_dist_fig}) represented by the undirected graph $\mathcal{G}$ which includes a set of nodes (vertices) $\mathcal{N}$ and a set of power lines (edges) $\mathcal{L}$. 
\begin{figure}[]
    \centering
    \includegraphics[width=0.8\columnwidth]{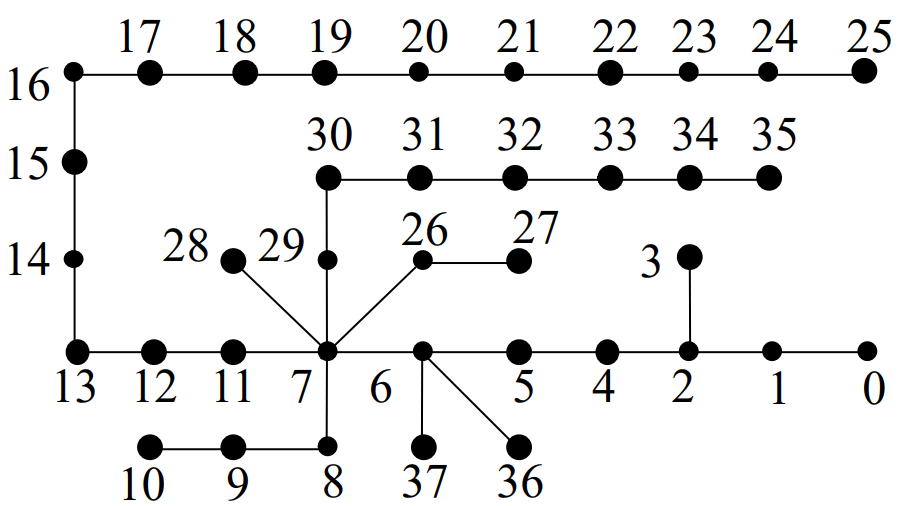}
    \caption{Radial distribution system.}
    \label{fig: radial_dist_fig}
\end{figure}
In this work, we consider each node as one population with its own daily load profile; however, each node could be an aggregation of smaller entities downstream of the local distribution connection point. 
The undirected graph is organized as a tree, with the root node representing the distribution system's substation where it is connected to the regional transmission system. We denote $N$ as the total number of nodes in the network excluding the root node. The nodes are indexed as $i=0,\dots,N$, and the node corresponding to $i=0$ (i.e., the root node) is the substation. The power lines are indexed by $i=1,\dots,N$ where the $i$-th line is directly upstream of node $i$ (i.e., line $i$ feeds directly to node $i$). In the following, we denote the parent vertex of node $i$ as $\mathcal{A}_i$ and the set of children vertices of node $i$ as $\mathcal{K}_i$.

Furthermore, we assume the aggregator has access to measurement data at each node's local connection point. Specifically, the aggregator measures the active power demands at each node $i$ at time $t$ on day $\tau$ denoted as $d_{i,\tau}^P(t)$. In order to ensure the delivered power is suitable for the electricity customers, the aggregator also monitors node $i$'s local voltage at time $t$ on day $\tau$ denoted as $v_{i,\tau}(t)$. 
In the following, we denote the active power daily load profile of node $i$ on day $\tau$ as $\mathbf{D}_{i,\tau}^P = [d_{i,\tau}^P(t)]_{t=1,\dots,T}$. Additionally, the aggregator records the active power flows $f_{i,\tau}^P(t)$ on each line $i\in\mathcal{L}$. \textcolor{black}{We note that reactive power should also be monitored in distribution systems, even though it is generally not priced and customers do not consider it in determining their optimal load response to prices. As such, we use the superscript $Q$ for the reactive power at a node, $d_{i,\tau}^Q(t)$, and for reactive power flow on a line, $f_{i,\tau}^Q(t)$.} Each line in the distribution system has its own internal resistance denoted as $R_i$, reactance denoted as $X_i$, and power limit denoted as $S_i^{max}$. The parameters for the distribution system are listed in Table \ref{Dist_table}.
\begin{table}
\begin{center}
\begin{footnotesize}
 \begin{tabular}{||p{1.5ex} p{6ex} p{6ex} p{5.5ex}|p{1.2ex} p{6ex} p{6ex} p{5.5ex}||} 
 \hline
 Line & \hfil R & \hfil X & S$^{max}$ & Line & \hfil R & \hfil X & S$^{max}$\\ [0.2ex] 
 & ($10^{-3} \Omega$) \hspace{-20ex}&($10^{-3} \Omega$) &(KVA) & & ($10^{-3} \Omega$) &($10^{-3} \Omega$) &(KVA)\\
 \hline\hline
 1 & 24.2 & 48.2 & 54 & 20 & 129.5 & 30.9 & 10.8\\ 
 \hline
 2 & 227.3 & 743.5 & 84 & 21 & 15.1 & 5.4 & 14.4 \\
 \hline
 3 & 76.3 & 18.2 & 10.8 & 22 & 50.8 & 12.1 & 10.8\\
 \hline
 4 & 43.6 & 142.7 & 84 & 23 & 69.1 & 16.5 & 10.8 \\
 \hline
 5 & 25.8 & 84.4 & 84 & 24 & 31.6 & 11.2 & 14.4\\
 \hline
 6 & 10.5 & 10.7 & 40.2 & 25 & 96.3 & 23 & 10.8\\ 
 \hline
 7 & 23.2 & 23.6 & 40.2 & 26 & 110.7 & 112.6 & 40.2\\ 
 \hline
 8 & 75.1 & 26.7 & 14.4 & 27 & 2.1 & 0.7 & 14.4\\ 
 \hline
 9 & 114.4 & 27.3 & 10.8 & 28 & 242.1 & 86.2 & 14.4\\ 
 \hline
 10 & 110.8.3 & 67.7 & 14.4 & 29 & 27.3 & 27.8 & 40.2\\ 
 \hline
 11 & 63.7 & 22.7 & 14.4 & 30 & 174.6 & 62.1 & 16.2\\ 
 \hline
 12 & 278.7 & 99.2 & 14.4 & 31 & 43 & 15.3 & 10.8\\ 
 \hline
 13 & 254.2 & 10.8.5 & 14.4 & 32 & 207.8 & 74 & 10.8\\ 
 \hline
 14 & 21.8 & 5.2 & 10.8 & 33 & 109.4 & 38.9 & 14.4\\ 
 \hline
 15 & 57.3 & 20.4 & 14.4 & 34 & 50.5 & 18 & 14.4\\ 
 \hline
 16 & 126.7 & 45.1 & 14.4 & 35 & 165.2 & 58.8 & 14.4\\ 
 \hline
 17 & 48.6 & 11.6 & 10.8 & 36 & 49.5 & 17.6 & 14.4\\ 
 \hline
 18 & 95.1 & 22.7 & 10.8 & 37 & 5.8 & 2.1 & 14.4\\ 
 \hline
 19 & 137.3 & 32.8 & 10.8 & & & & \\ [0.2ex] 
 \hline
\end{tabular}
\caption{Distribution system parameters.}
\label{Dist_table}
\end{footnotesize}
\end{center}
\end{table}

\subsection{Power Flow Model}
\label{subsection: Power Flow Model}
In order to solve for the power flow and nodal voltages of the power distribution system, we make use of the \textit{LinDistFlow} model\cite{LinDistFlow_Original}, which is a linear approximation for the AC power flow model\footnote{\textcolor{black}{The reader should note that the proposed learning approach is not limited to the \textit{LinDistFlow} model. There are other power flow models that can be utilized such as \cite{sankur2016linearized}.}}. The \textit{LinDistFlow} model has been extensively studied and verified to be competitive to the nonlinear AC flow model on many realistic feeder topologies including radial \cite{liu2017decentralized_thesis, zhu2016fast_lindistflow_compare, vsulc2014optimal_lindistflow_compare, farivar2013equilibrium_lindistflow_compare}. The \textit{LinDistFlow} model reduces computational complexity by making use of the following linear power flow and voltage equations:

\small\begin{align}
    \label{eqn: active_flow}
    d_{i,\tau}^P(t)  + \sum_{j\in\mathcal{K}_i} f_{j,\tau}^P(t) = f_{\mathcal{A}_i,\tau}^P(t); \;\;&\forall t,\tau,  i,\\
    \label{eqn: reactive_flow}
    d_{i,\tau}^Q(t)  + \sum_{j\in\mathcal{K}_i} f_{j,\tau}^Q(t) = f_{\mathcal{A}_i,\tau}^Q(t); \;\;&\forall t,\tau, i,\\
    \label{eqn: voltage_drop}
    u_{\mathcal{A}_i,\tau}(t) - 2\big(f_{i,\tau}^P(t)R_i + f_{i,\tau}^Q(t)X_i\big) = u_{i,\tau}(t); \;&\forall t,\tau, i
\end{align}

\normalsize
\noindent \textcolor{black}{where \eqref{eqn: active_flow} accounts for active power and \eqref{eqn: reactive_flow} accounts for reactive power.} In \eqref{eqn: voltage_drop} we make use of the operator $u_{i,\tau}(t) = \big(v_{i,\tau}(t)\big)^2$ to provide a linear voltage drop relationship across the distribution system. For the scope of this work, we assume that the substation connection to the regional transmission system (node $i=0$) is regulated and has a fixed voltage $v_{0,\tau}(t) = \textcolor{black}{12.5  \textrm{kV}}, \forall t,\tau$.

\subsection{Distribution System Operational Constraints}
\label{subsection: Network Constraints}
The nodal voltages and line flows calculated in \eqref{eqn: active_flow}-\eqref{eqn: voltage_drop} should obey the following constraints for reliable operation:
\begin{align}
    \label{eqn: voltage min const exp}
    u_{i,\tau}(t) &\geq u_i^{min}, &&\forall t,\tau,i\in\mathcal{N},\\
    \label{eqn: voltage max const exp}
    u_{i,\tau}(t) &\leq u_i^{max}, &&\forall t,\tau,i\in\mathcal{N},\\
    \label{eqn: power flow const exp}
    f_{i,\tau}^P(t)^2 +f_{i,\tau}^Q(t)^2&\leq (S_i^{max})^2, &&\forall t,\tau,i\in\mathcal{L},
\end{align}
where \eqref{eqn: voltage min const exp}-\eqref{eqn: voltage max const exp} are the nodal voltage constraints and \eqref{eqn: power flow const exp} corresponds to the power constraints for each distribution line.
\begin{figure*}[ht!]
    \centering
    \includegraphics[width=1.0\textwidth]{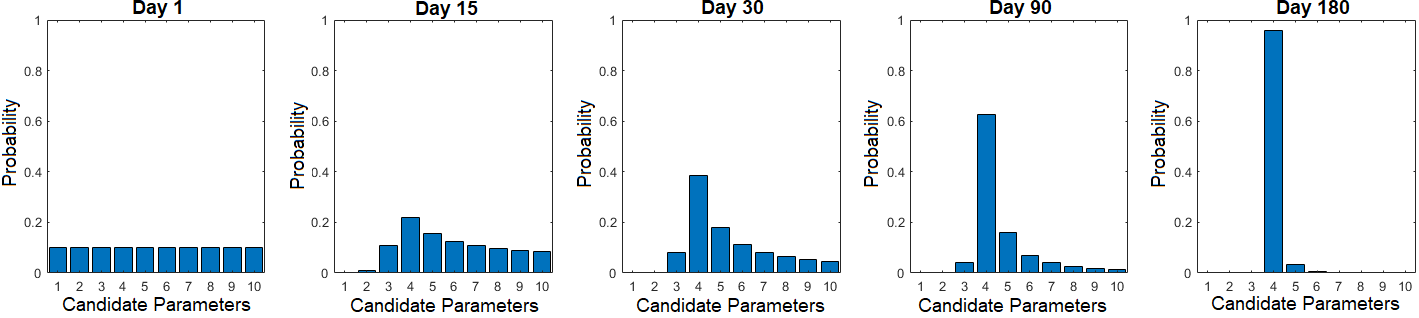}
    \caption{\textcolor{black}{The 5 plots above portray the evolution of the aggregator's knowledge of the population's hidden parameter at node 10 throughout the learning procedure. The true parameter is parameter 4. From left to right: Day 1 (initialized to uniform distribution, i.e., no knowledge of the true parameter), Day 15 prior, Day 30 prior, Day 90 prior, and Day 180 prior. At day 180, the aggregator is about 95\% certain that parameter 4 is the true parameter. }}
    \label{fig: prior}
\end{figure*}
\subsection{Load Model and Multi-armed Bandit Formulation}
\label{section initial sim}
In this test case, we consider 6 time slots each day, each 4 hours long. We consider 10 unique target load profile vectors, with the daily target profile $\mathbf{V}_{\tau}$ for day $\tau$ drawn from a uniform distribution each morning. Each of the 10 target load profile vectors corresponds to a desired load curve to accommodate different levels of forecasted renewable generation. Furthermore, the aggregator transmits daily price signals $\mathbf{p}_{i,\tau}$ to each node within the system. The aggregator has a high and low price for each of the 6 time slots resulting in $2^6$ possible daily price signals. Since the aggregator is shaping the electricity demand at each node within the distribution system, each node has its own cost $f\big(\mathbf{D}_{i,\tau}^{}(\mathbf{p}_{i,\tau}), \mathbf{V}_{\tau}\big)$ that is dependent on the node's daily demand and the target profile. \textcolor{black}{In this test case, we assume  the cost function for each node is the squared deviation of the node's electricity demand from the target profile: $f\big(\mathbf{D}_{i,\tau}^{}(\mathbf{p}_{i,\tau}), \mathbf{V}_{\tau}\big)=|\mathbf{D}_{i,\tau}^{}(\mathbf{p}_{i,\tau})- \mathbf{V}_{\tau}|^{2}$, thus equally penalizing over-usage and under-usage of electricity. We note that the units are KW$^2$ and if the aggregator had a converting function for the squared deviation (KW$^2$) to $\$$U.S.D., then we could calculate the monetary losses of the system. In our experimental examples, we make use of discrete sets for the available $\theta$'s and $\mathbf{p}_{\tau}$'s to guarantee that an enumeration-based method could solve for the globally optimal price signals each day in spite of problem non-convexities.}

We consider 20 unique load flexibility clusters in this test case. \textcolor{black}{Each cluster's parameters represent the varying start/stop times, total energy demands, and power limitations common to EV loads in residential areas and are of the form presented in equation \eqref{eqn: cluster_example}.} \textcolor{black}{We note that we generated the population’s load price response directly using the same clustering model (i.e., the actual load response in the simulation is at the level of 20 clusters and can be well represented by the 20 clusters \textit{plus additive noise}. For a discussion on the effects of poor clustering, we refer the reader to Section \ref{section: effects of clustering}).} Each node in the distribution system is comprised of these 20 load clusters with its own unique sensitivities $a_{i,c}(\mathbf{p}_\tau)$ for each cluster. Each sensitivity parameter is selected as $a_{i,c}(\mathbf{p}_{i,\tau}) \sim \mathcal{N}(\frac{\beta_c}{\boldsymbol{\theta}_i^{\star}\mathbf{p}_{i,\tau}},\sigma^2)$ each day where $\beta_c$ is a cluster specific constant known by the aggregator \textcolor{black}{(we note that $\beta_c$ represents \textit{a priori} knowledge of customers' preferences and could come from behavioral studies; however, our framework does not require this and $\beta_c$ can be completely omitted in cases where prior information is unavailable)}. Each node's price sensitivity, i.e., parameter to be learned, $\boldsymbol{\theta}_i^{\star}$, is a vector of length 6 and the set of possible parameters, $\Theta$, contains 10 unique vectors. \textcolor{black}{Unless noted, the reliability parameter chosen for the Con-TS-RTP algorithm is $\nu=0.1$.}

\textcolor{black}{\textit{Note on reactive power}: We note that reactive power is generally price insensitive; however, reactive power is present in a distribution system and affects the constraints of the system. Reactive power flows alter how the price sensitive loads are limited by the operational constraints of the system (i.e., active and reactive flows on lines affect the capacity available for the price responsive loads). Due to  the lack of data as to how much reactive power is present in the distribution system due to our appliance clusters and otherwise, for our numerical examples, we omit the inclusion of reactive power to only view the appliance clusters’ active load profiles within the distribution system. For further discussion on this, we refer the reader to papers that fully capture the effects of reactive power in such problems such as \cite{8600344} and \cite{8910409} in which the authors showcase techniques to handle distribution systems with chance constraints.} 

\textcolor{black}{In the following sections, without loss of generality, we assume that reactive power is not responsive to the pricing signals. We note that our proposed learning approach can accommodate reactive power flows (\textit{LinDistFlow} can as well); however, our goal was to show proof of concept of our learning/pricing approach with active customer loads, thus reactive power flow will be examined in future work.}

\subsection{Results}
\label{section: original sim}

We simulated the Con-TS-RTP algorithm for 365 days for an aggregator attempting to learn the sensitivities of the nodes in the system and shape their demands. In the following, we highlight the results of the simulation at node 10 of the radial distribution system. Figure \ref{fig: prior} presents the evolution of the prior distribution for node 10's hidden parameter. On day 1, the prior was initialized to a uniform distribution among the candidate parameters, and by day 180 the weight on the true parameter exceeded $0.95$.

Figure \ref{fig: regret} presents the regret performance of Con-TS-RTP at node 10. As seen in Figure \ref{fig: regret}, the regret curve flattens after day 130 as the algorithm never chooses a suboptimal price signal after this day. 
\begin{figure}[h]
    \centering
    \includegraphics[width=1.0\columnwidth]{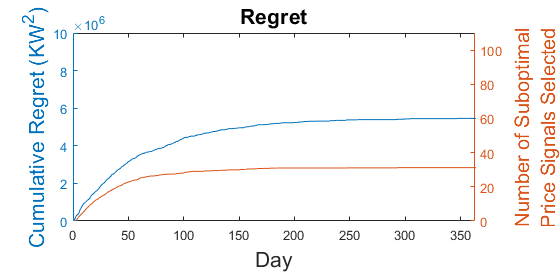}
    \caption{Regret performance of Con-TS-RTP at node 10 \textcolor{black}{with $\nu=0.1$. Note that the y-axis (left) units are KW$^2$ for the squared load deviation from the target profile.}}
    \label{fig: regret}
\end{figure}

Figure \ref{fig: deviation} presents node 10's deviation from a specific daily target profile. On days 2, 3, 4, 53, and 365 the same target profile (i.e.,  $\mathbf{V}_2=\mathbf{V}_3=\mathbf{V}_4=\mathbf{V}_{53}=\mathbf{V}_{365}$) was drawn and the aggregator selected different price signals to shape the node's demand. As seen in Fig. \ref{fig: deviation}, the deviation from the target profile on day 365 is less than the deviation on the other days as the algorithm has learned the true parameter and selects the optimal price signal to shape the load.
\begin{figure}[h]
    \centering
    \includegraphics[width=0.8\columnwidth]{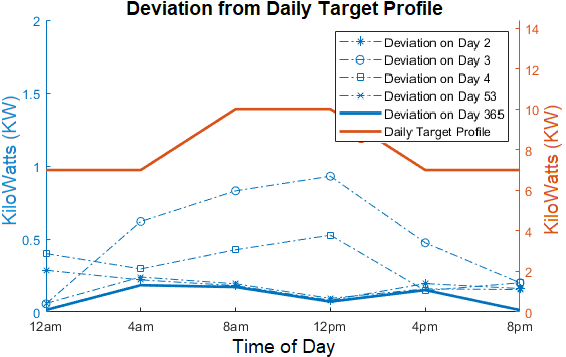}
    \caption{Deviation of node 10's demand from a specific daily target profile.}
    \label{fig: deviation}
\end{figure}

In Figure \ref{fig: violations}, we present the distribution system constraint violations that were avoided by using Con-TS-RTP instead of an unconstrained TS algorithm. Clearly, in the early learning stages, the unconstrained TS algorithm does not have accurate knowledge of the hidden parameters and violates the distribution system constraints often. Con-TS-RTP is more conservative with its exploration of untested price signals and avoids the constraint violations made by the unconstrained TS algorithm. Last, we note that the simulation was implemented in Matlab and CVX on an i7 processor with 16gb of RAM. The 365 day simulations were run in less than 5 minutes.

\begin{figure}[h]
    \centering
    \includegraphics[width=0.8\columnwidth]{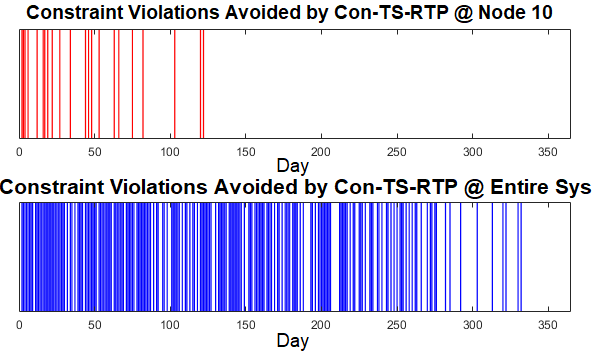}
    \caption{Top: Distribution system constraint violations at node 10 avoided by using Con-TS-RTP instead of an unconstrained TS. Bottom: Distribution system constraint violations across the entire system avoided by using Con-TS-RTP instead of an unconstrained TS.}
    \label{fig: violations}
\end{figure}

\textcolor{black}{
\subsection{Effects of Clustering}
\label{section: effects of clustering}
In this section, we portray the effects of selecting different numbers of clusters to represent a true load as well as the effects of selecting too few clusters on the performance of our Con-TS-RTP algorithm. First, in Figure \ref{fig: inaccurate clusters loads}, we  perform a simple demonstration. We considered a population of 100 EVs with random charging requests and then constructed clusters to view the accuracy of the clustered load profiles versus the actual load profile. As shown in Figure \ref{fig: inaccurate clusters loads}, using 1, 5, or 10 clusters to represent the EV population results in load profiles quite different from the actual; however, with 20 clusters, the load profile begins to match the actual profile. 
\begin{figure}[h]
    \centering
    \includegraphics[width=1.0\columnwidth]{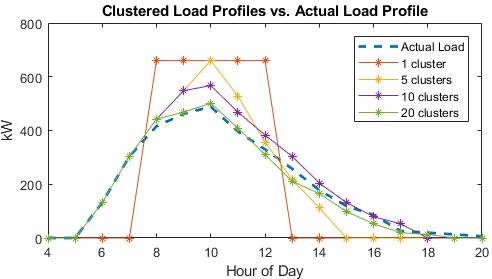}
    \caption{\textcolor{black}{Effects of changing the number of clusters to model an actual load. Specifically, load profiles for 4 cluster models compared to the actual load profile for a population of 100 charging EVs.}}
    \label{fig: inaccurate clusters loads}
\end{figure}
}

\textcolor{black}{
Furthermore, in Figure \ref{fig: inaccurate cluster regret} we show the effects of reducing the number of clusters in the load model on the regret performance of our Con-TS-RTP algorithm. Specifically, we focus on the same setup as Section \ref{section initial sim} with the exception that we have the Con-TS-RTP algorithm use a 10 cluster model instead of the 20 cluster model for the population to see the effects of an inaccurate cluster model. 
As shown in Figure \ref{fig: inaccurate cluster regret}, the regret curve for this case never flattens and the algorithm is never able to select the optimal price signal. This is because the algorithm's model of the load (i.e., the 10 clusters) is unable to accurately model the population’s response and causes the algorithm to select incorrect prices every day. 
\begin{figure}[]
    \centering
    \includegraphics[width=0.8\columnwidth]{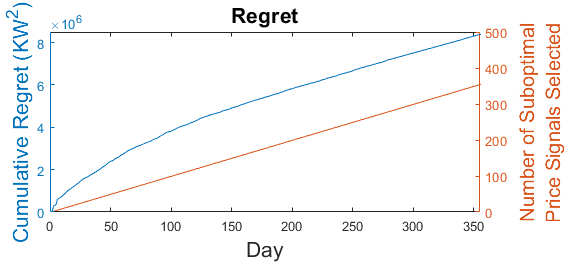}
    \caption{\textcolor{black}{Effects of using too few clusters for the population’s load model. We show the regret performance of Con-TS-RTP at node 10 with $\nu=0.1$ for a 10 cluster model instead of the 20 cluster model as previously shown in Fig. \ref{fig: regret}. Due to the inaccuracies of the 10 cluster model, the algorithm is never able to select the optimal price signals.}}
    \label{fig: inaccurate cluster regret}
\end{figure}
}

\textcolor{black}{
\subsection{Evolving Price Sensitivity}
\label{section: evolve}
In this section, we show an example of what happens when customers' sensitivities change over time and how a Bayesian learning approach can naturally adapt and account for these dynamic changes. Specifically, we simulated the same system as in Figure \ref{fig: regret}, but on day 250, we altered the true $\boldsymbol{\theta}^{\star}_i$ parameter. As seen in Figure \ref{fig: evolve}, the regret curves first flatten around day 125, then increase at day 250, and then flatten again near day 325. This shows that Con-TS-RTP was able to successfully learn the first and second true parameter without any modifications to the algorithm. The algorithm naturally shifts its belief about the true parameter as it observes outputs that do not match its current belief.
\begin{figure}[h]
    \centering
    \includegraphics[width=1.0\columnwidth]{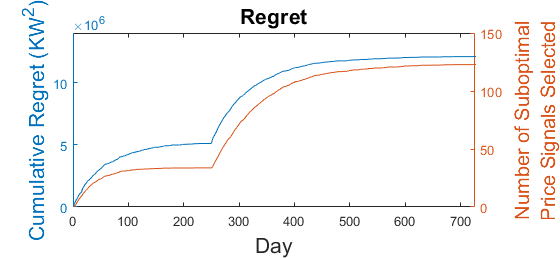}
    \caption{\textcolor{black}{Regret performance of Con-TS-RTP at node 10 with $\nu=0.1$. Note that on day 250, the hidden parameter was altered.}}
    \label{fig: evolve}
\end{figure}
}

\textcolor{black}{
\subsection{Non-repeating Target Profiles}
\label{section: evolve}
In the previous case study, we assumed a low number of target profiles (10 profiles) to satisfy the assumptions we have made for our theoretical results.
In this section, we demonstrate how extending the number of target profiles to 365 does not negatively affect the performance of the algorithm in practice. Furthermore, we ensure that once a target profile has been viewed by the aggregator, it is never drawn again. Thus, each day the aggregator is posting a price to shape the population's load to match a target profile that it has never seen before. As shown in Figure \ref{fig: exogenous_enlarge}, enlarging the set of target profiles does not slow down the learning process. Note that in Figure \ref{fig: exogenous_enlarge} the regret flattens near trial 100 which matches the duration of the learning period seen in Figure \ref{fig: regret} (i.e., in simulation, the aggregator is still able to learn the true parameter when the number of target profiles is increased from 10 to 365, resulting in similar regret curves).
\begin{figure}[h]
    \centering
    \includegraphics[width=\columnwidth]{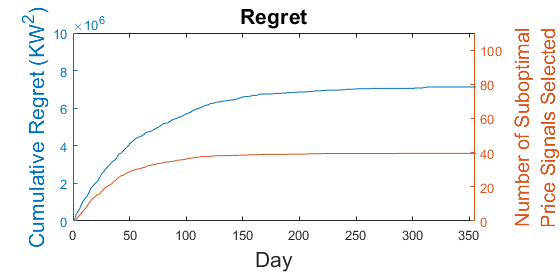}
    \caption{\textcolor{black}{Regret performance of Con-TS-RTP at node 10 with $\nu=0.1$. Note that on each day, the sampled $\mathbf{V}_{\tau}$ has never been seen by the aggregator.}}
    \label{fig: exogenous_enlarge}
\end{figure}
}

\textcolor{black}{
\subsection{Effects of Varying the System Reliability Metric}
\label{results-varying reliability}
In this section, we discuss the effects of varying the system reliability parameters in the daily optimization's constraints (i.e., altering the value of $\nu$ for the system constraints formulated as in \eqref{prob1}-\eqref{prob3}). As described in Sections \ref{subsection: Cluster Price Response} and \ref{subsection: Thompson Sampling 2}, the reliability metric dictates the aggregator's allowable probability of a constraint violation under its current belief distribution about the unknown parameter. Decreasing $\nu$ is restricting the algorithm to avoid violations and setting $\nu=1$ is equivalent to solving the daily optimization without the constraints altogether. In Figure \ref{fig: varyingSafety}, we simulated the system with varying reliability parameters. Specifically, each curve shown is the average regret at node 10 over 20 independent simulations. As shown in Figure \ref{fig: varyingSafety}, the regret increases as the desired reliability increases (smaller $\nu$). This is because the aggregator is forced to select more conservative prices during the learning procedure to ensure that the constraints are met with higher probability. 
}

\begin{figure}[h]
    \centering
    \includegraphics[width=0.9\columnwidth]{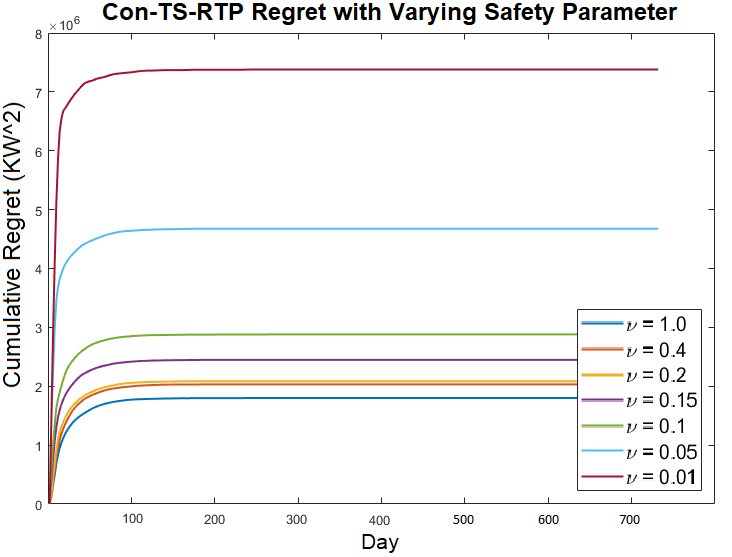}
    \caption{\textcolor{black}{Regret curves for various system reliability metrics. Each curve is an average of 20 independent simulations.}}
    \label{fig: varyingSafety}
\end{figure}

\textcolor{black}{
\subsection{Comparison with Two-Stage Learning}
\label{section: twostage}
In this section, we present a comparison of the Con-TS-RTP approach versus a 2-stage ``learn" and then ``optimize" algorithm, where the first stage consists of pure exploration and the second stage purely exploits the knowledge gained in the first stage. The simulation setup is the same as the setup used in Section \ref{section initial sim}. A description of the 2-stage algorithm used is as follows: The aggregator decides the duration of the learning stage \textit{a priori}, (in Figure \ref{fig: twostage}, we present regret curves for learning stages with durations of 5, 15, and 25 days) and during this learning stage, the aggregator randomly selects price signals from a predetermined safe set of prices (i.e., prices high enough such that constraints cannot be violated), observes the populations' responses, and performs posterior updates. Then, after the learning stage is complete, for the remainder of time the aggregator broadcasts the best price signals with respect to the knowledge of the unknown parameter at the end of the learning stage (the selected price signal will ensure safety but might be potentially suboptimal depending on the duration of the learning stage). Clearly, the two most significant shortcomings of the 2-stage approach are: 1) arbitrarily bad performance during the learning stage due to random price selection; and 2) difficulty selecting a sufficient duration of the learning stage. As seen in Figure \ref{fig: twostage}, this 2-stage myopic algorithm results in linear regret in the 5 day and 15 day learning stage curves. Due to an insufficient number of posterior updates, the aggregator is forced to post suboptimal price signals to ensure safety given its noisy knowledge of the unknown parameter after the learning stage is over. On the other hand, the 25 day learning stage is able to converge to the optimal price signals, but the performance during the learning stage causes fast growth of regret whereas Con-TS-RTP is able to avoid all of the aforementioned shortcomings.
}

\begin{figure}[h]
    \centering
    \includegraphics[width=\columnwidth]{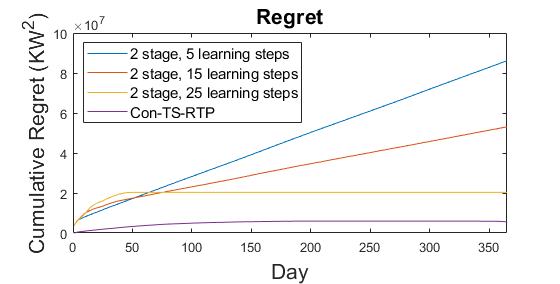}
    \caption{\textcolor{black}{Regret performance of Con-TS-RTP and a 2-Stage algorithm at node 10 with $\nu=0.1$. Note that the 5 day (blue) and 15 day (red) learning algorithms were unable to converge to the optimal price signals. \textcolor{black}{The blue and red curves never flatten because their learning stages were too brief to adequately learn the customers’ preferences and are unable to select the optimal price signals, resulting in a linearly growing regret. However, the 25 day (yellow) learning stage algorithm is able to adequately learn the population’s parameters and select optimal prices after that.}}}
    \label{fig: twostage}
\end{figure}

\section{Conclusion}
In this paper, we presented a multi-armed bandit problem formulation for an electricity aggregator attempting to run a real-time pricing program for load shaping (e.g., to reduce demand at peak hours, integrate more intermittent renewables, track a desired daily load profile, etc). We made use of a constrained Thompson sampling heuristic, Con-TS-RTP, as a solution to the \textit{exploration/exploitation} problem of an aggregator \textit{passively} learning customers' price sensitivities while broadcasting price signals that influence customers to alter their demand to match a desired load profile. The proposed Con-TS-RTP algorithm permits day-varying target load profiles (i.e., multiple target load profiles reflecting renewable forecasts and desired demand patterns) and takes into account the actual operational constraints of a distribution system to ensure that the customers receive adequate service and to avoid potential grid failures. Additionally, our setup accounts for complex electricity usage patterns of the customers by classifying different load clusters based on electricity demand and load flexibility. We discussed a regret guarantee for the proposed Con-TS-RTP algorithm which bounds the total number of suboptimal price signals broadcasted by the aggregator. Furthermore, we discussed an operational reliability guarantee that ensures the power distribution system constraints are upheld with high probability throughout the run of the Con-TS-RTP algorithm.


\bibliographystyle{IEEEtran}
\bibliography{references}

\section*{Appendix}
\vspace{-20pt}
\textcolor{black}{
\subsection{Table of Notation}
\vspace{-6pt}
\begin{align}
    &{\tau} &&\text{Day index} \nonumber \\[-2.0pt]
    &\mathcal{T} &&\text{Total number of days} \nonumber \\[-2.0pt]
    &t &&\text{Time of day index} \nonumber \\[-2.0pt]
    &T &&\text{Number of time epochs in a day} \nonumber \\[-2.0pt]
    &\mathbf{p}_{\tau} &&\text{Daily dispatch signal} \nonumber \\[-2.0pt]
    &\mathcal{P} &&\text{Set of available dispatch signals}\nonumber \\[-2.0pt]
    &f(\cdot) &&\text{Aggregator's fixed daily cost function}\nonumber\\[-2.0pt]
    &\mathbf{D}_{\tau}(\mathbf{p}_{\tau}) &&\text{Population's daily load response} \nonumber \\[-2.0pt]
    &\mathbf{V}_{\tau} &&\text{Daily exogenous parameter} \nonumber \\[-2.0pt]
    &\mathcal{V} &&\text{Set of exogenous parameters} \nonumber \\[-2.0pt]
    &g(\cdot) &&\text{General reliability constraint} \nonumber \\[-2.0pt]
    &u_{\tau}(t) &&\text{Nodal voltage at time }t\text{ on day }\tau  \nonumber \\[-2.0pt]
    &f_{\tau}(t) &&\text{Power flow at time }t\text{ on day }\tau  \nonumber \\[-2.0pt]
    &c &&\text{Flexible appliance cluster index} \nonumber \\[-2.0pt]
    &\mathcal{C} &&\text{Set of flexible appliance clusters} \nonumber \\[-2.0pt]
    &\mathbf{D}_{c} &&\text{Load profile for cluster }c \nonumber \\[-2.0pt]
    &\mathcal{D}_c &&\text{Set of load profiles for cluster }c \nonumber \\[-2.0pt]
    &E_c &&\text{Total energy required by cluster }c \nonumber \\[-2.0pt]
    &\rho_c &&\text{Maximum power rating for cluster }c \nonumber \\[-2.0pt]
    &a_c(\mathbf{p}_{\tau}) &&\text{Preference adjustment model for cluster }c \nonumber \\[-2.0pt]
    &\mathcal{D} &&\text{Set of load profiles for entire population} \nonumber \\[-2.0pt]
    &\widetilde{\mathbf{D}}_{c,\tau}(\mathbf{p}_{\tau}) &&\text{Minimum cost load profile of cluster }c \nonumber \\[-2.0pt]
    &\mathbf{D}_{\tau}^{\star}(\mathbf{p}_{\tau}) &&\text{Population's realized load profile on day }\tau \nonumber \\[-2.0pt]
    &\phi_c &&\text{Preference adjustment distribution for cluster }c \nonumber \\[-2.0pt]
    &\boldsymbol{\theta}^{\star} &&\text{True customer sensitivity model} \nonumber \\[-2.0pt]
    &\Theta &&\text{Set of candidate sensitivity models} \nonumber \\[-2.0pt]
    &\mu &&\text{Aggregator's desired reliability metric} \nonumber \\[-2.0pt]
    &\nu &&\text{Alternate reliability metric} \nonumber \\[-2.0pt]
    &R_{\mathcal{T}} &&\text{Cumulative regret after day }\mathcal{T} \nonumber \\[-2.0pt]
    &N_{\tau}(\mathbf{p},\mathbf{V}) &&\text{Number of times price }\mathbf{p} \text{ selected in response }\nonumber \\[-2.0pt]
    & &&\text{to vector }\mathbf{V} \text{ up to day }\tau \nonumber \\[-2.0pt]
    &\pi_{\tau} &&\text{Prior distribution (aggregator's belief on }\boldsymbol{\theta}^{\star}) \nonumber \\[-2.0pt]
    &\widetilde{\boldsymbol{\theta}}_{\tau} &&\text{Sampled parameter on day }\tau \nonumber \\[-2.0pt]
    &\mathbf{Y}_{\tau} &&\text{Observed load profile on day }\tau \nonumber \\[-2.0pt]
    &\mathcal{Y} &&\text{Observation space} \nonumber \\[-2.0pt]
    &N &&\text{Number of nodes in distribution system} \nonumber \\[-2.0pt]
    &\mathcal{A}_i &&\text{Parent vertex of node }i \nonumber \\[-2.0pt]
    &\mathcal{K}_i &&\text{Set of children vertices of node }i \nonumber \\[-2.0pt]
    &R_i &&\text{Internal resistance of line }i \nonumber \\[-2.0pt]
    &S_i^{max} &&\text{Maximum power limit of line }i \nonumber \\[-2.0pt]
    &d_{i,\tau}(t) &&\text{Power demand at node }i \text{ at time }t \text{ on day }\tau\nonumber \\[-2.0pt]
    &\beta_c &&\text{Prior knowledge of preferences of cluster } c \nonumber
\end{align}
}
\vspace{-10pt}
\subsection{Discussion on Regret Performance}
In this section, we describe the necessary background for Theorem \ref{thm conTS} and then present the full version of the Theorem. Recall, $\mathbf{p}^{\mathbf{V}_{\tau},\star}$ denotes the optimal price signal for the true model of the population's price response $\boldsymbol{\theta}_{}^{\star}$ when the daily exogenous parameter $\mathbf{V}_{\tau}$ is observed on day $\tau$. Any price signal $\mathbf{p}_{\tau}\neq\mathbf{p}_{}^{\mathbf{V}_{\tau},\star}$ is considered a suboptimal price. 

We now briefly explain how the posterior updates affect the regret performance. When price ${\bf  p}$ is posted on day $\tau$, the prior density is updated  as 
\begin{equation}\label{updateprior}
\pi_{\tau}(d \boldsymbol{\theta}) \propto \exp\left( - \log \frac{l( \mathbf{Y}_{\tau};{\bf  p},\boldsymbol{\theta}^{\star})}{l( \mathbf{Y}_{\tau};{\bf  p},\boldsymbol{\theta})}\right) \pi_{\tau-1}(d\boldsymbol{\theta}). \end{equation}
Now, denote by $\textit{KL}(\boldsymbol{\theta}_{\bf p} ^{\star} || \boldsymbol{\theta}_{\bf p})$  the marginal Kullback-Leibler divergence between the distribution $\{ l(Y;{\bf  p},\boldsymbol{\theta}^{\star}): Y\in \mathcal{Y}\} $ and $\{ l(Y;{\bf  p},\boldsymbol{\theta}):Y \in \mathcal{Y}\} $. 
As   in \cite{gopalan2014thompson}, we can approximately write \eqref{updateprior} as:
\begin{align}\label{Approx.dsit}
    \pi_{\tau}(d \boldsymbol{\theta}) \propto \exp \bigg(- \sum_{\mathbf{p} \in \mathcal{P}}N_{\tau}({\bf  p}) \textit{KL}(\boldsymbol{\theta}_{{\bf  p}} ^{\star} ||  \boldsymbol{\theta}_{{\bf  p}}) \bigg) \pi_{\tau-1}(d \boldsymbol{\theta}),
\end{align}
where $N_{\tau}({\bf  p}) = \sum_{{{\bf  V}} \in \mathcal{V}} N_{\tau}({\bf  p},{\bf  V})$, and   $N_{\tau}({\bf  p},{\bf  V})$ is the number of times up to  day $\tau$ that the algorithm simultaneously observes a   daily target load profile   ${\bf  V}$ and posts a price ${\bf  p}$. As such, the total number of times that suboptimal price signals are selected over $\mathcal{T}$ days is:
\begin{align}
    \sum_{\mathbf{V}\in\mathcal{V}}
    \sum_{\mathbf{p}\in\{\mathcal{P}\setminus\mathbf{p}^{\mathbf{V},\star}\}}
    N_{\mathcal{T}}(\mathbf{p},\mathbf{V})
    =
    \sum_{\tau=1}^{\mathcal{T}}
    \mathbbm{1}^{\{\mathbf{p}_{\tau}\neq\mathbf{p}_{}^{\mathbf{V}_{\tau},\star}\}},
\end{align}
where $\mathbbm{1}^{\{\cdot\}}$ is the indicator function that is set equal to one if the criteria is met and zero otherwise.

Furthermore, we define ${\bf N}_{\tau} = [N_{\tau}({\bf p})]_{{\bf p} \in \mathcal{P}}$ as a vector consisting of the number of times each  price is posted up to day $\tau$.  We can consider the quantity in the exponent of \eqref{Approx.dsit} as a loss suffered  by model $\boldsymbol{\theta}$ up to day $\tau$.
Since the term in the exponent of \eqref{Approx.dsit} is equal to 0 when $\boldsymbol{\theta} = \boldsymbol{\theta}^{\star}$, we can see that Thompson sampling samples $\boldsymbol{\theta}^{\star}$ and hence posts the optimal price with at least a constant probability at each day, i.e., $N_{\tau}(\mathbf{p}^{\mathbf{V}_{},\star},{\bf  V_{}})$ grows linearly with $\tau$ for all 
${\bf  V}$.  

For each price, we define $S_{{\bf  p}}({\bf  V}) := \{ \boldsymbol{\theta} \in \Theta: {\bf  p}_\tau = {\bf  p}|{\bf V}_\tau = {\bf V} \}$ to be the set of parameters $\boldsymbol{\theta} \in \Theta$ whose  optimal price when observing a daily target load profile   $\bf V$ is ${\bf  p}$. Furthermore, define $S_{\textbf p}^{'}({\bf  V}) := \{ \boldsymbol{\theta} \in S_{\textbf p}({\bf  V}): \textit{KL}(\boldsymbol{\theta}_{{\bf  p}^{{\bf V,\star}}}^{\star} \| \boldsymbol{\theta}_{{\bf  p}^{{\bf V,\star}}} ) = 0\} $ which is the set of  models $\boldsymbol{\theta}$ that exactly match $\boldsymbol{\theta}^{\star}$ in marginal distribution of  $Y$ when the true model $\boldsymbol{\theta}^{\star}$ is selected and  the optimal price ${\bf  p}^{{\bf V,\star}}$ is posted, and $S_{{\bf  p}}^{''} ({\bf  V}) := S_{{\bf  p}}({\bf  V}) \backslash S_{\textbf p}^{'}({\bf  V})$.

For each of the models $\boldsymbol{\theta}$ in $S_{{\bf  p}}^{''}({\bf  V})$, ${\bf  p} \neq {\bf  p}^{{\bf V,\star}}, \textit{KL}(\boldsymbol{\theta}_{{\bf  p}^{{\bf V,\star}}}^{\star} \| \boldsymbol{\theta}_{{\bf  p}^{{\bf V,\star}}} ) > \varepsilon >0$. As we have assumed that the probability of observing  any target profile ${\bf V} \in \mathcal V$ is bounded away from zero, $N_{\tau}({\bf  p}^{{\bf V,\star}})$   grows linearly with $\tau$ for all ${\bf V}  \in \mathcal V$. Hence, any  such model $\boldsymbol{\theta}$  is sampled with probability exponentially decaying in $\tau$ in \eqref{Approx.dsit} and the regret from such $S_{{\bf  p}}^{''}({\bf  V})$-sampling is negligible. We define the set of all such models as $\boldsymbol{\theta} \in \Theta'' = \cup_{{\bf  V} \in \mathcal V} S_{{\bf  p}}^{''}({\bf  V})$.

A model $\boldsymbol{\theta} \in S_{{\bf p}}^{'}({\bf V})$ will only face loss whenever the algorithm posted a suboptimal price ${\bf  p}$ for which $\textit{KL}(\boldsymbol{\theta}_{{\bf  p}}^{\star} \| \boldsymbol{\theta}_{{\bf  p}}) > 0$. For  ${\bf V}$,  a suboptimal price ${\bf  p}_k^{{\bf  V}} \neq {\bf  p}^{{\bf V,\star}}$ may still be posted if any of the set of models in $S_{{\bf  p}_k^{{\bf  V}}}^{'} ({\bf  V})$ may still be drawn with non-negligible probability. Hence, a price will be eliminated after the probability of drawing all $\boldsymbol{\theta} \in S_{{\bf p}_k^{{\bf  V}}}^{'}({\bf  V})$ is negligible. For each ${\bf V}$, suboptimal prices are eliminated one after the other at times   $\tau_k^{\bf V}, k = 1,\ldots,|\mathcal P|-1$. 
We refer the reader to \cite{gopalan2014thompson} for a full discussion of when a suboptimal price ${\bf p}$ is considered statistically eliminated, which is used to write constraints in \eqref{eqn: clogT} below.

\addtocounter{theorem}{-1}
\begin{theorem}
(Expanded Version) Under assumptions \ref{assump finite}-\ref{assump unique best action} and Constraint Set A in Algorithm \ref{alg: conTS}, for $\delta, \epsilon \in (0,1)$, there exists $\mathcal{T}^{\star}\geq0$ s.t. for all $\mathcal{T}\geq\mathcal{T}^{\star}$, with probability $1-\delta$:\\
\begin{align*}
    \sum_{\mathbf{V}\in\mathcal{V}}
    \sum_{\mathbf{p}\in\{\mathcal{P}\setminus\mathbf{p}^{\mathbf{V},\star}\}}
    N_{\mathcal{T}}(\mathbf{p},\mathbf{V})\leq B + C(\log\mathcal{T}),
\end{align*}

\noindent where $B\equiv B(\delta,\epsilon,\mathcal{P},\mathcal{Y},\Theta)$ is a problem-dependent constant that does not depend on $\mathcal{T}$, and $C(\log\mathcal{T})$ depends on $\mathcal{T}$, the sequence of selected price signals, and the Kullback-Leibler divergence properties of the bandit problem (i.e., the marginal Kullback-Leibler divergences of the observation distributions $\text{KL}\big[\ell(\mathbf{Y};\mathbf{p},\boldsymbol{\theta}^{\star}), \ell(\mathbf{Y};\mathbf{p},\boldsymbol{\theta})\big]$. Specifically, the $C(\log\mathcal{T})$ term is defined as follows:
\begin{align}
\label{eqn: clogT}
    &C(\log\mathcal{T})\equiv\\
    \nonumber
    &\max \sum_{\mathbf{V}\in\mathcal{V}} \sum_{k = 1}^{|\mathcal{P}|-1}  N_{\tau_k^{\mathbf{V}}}(\mathbf{p},\mathbf{V})\\
    \nonumber
    &\text{ s.t. }\forall \;\mathbf{V}\in\mathcal{V},\; \forall j>1,\; \forall1\leq k\leq |\mathcal{P}|-1 :\\
    \nonumber
    & \min_{\boldsymbol{\theta}\in\Big\{ S^{'}_{\mathbf{p}_k^{\mathbf{V}}}(\mathbf{V})-\Theta^{''}\Big\}} \langle \mathbf{N}_{\tau_k^{\mathbf{V}}}, \text{KL}_{\boldsymbol{\theta}} \rangle \geq \frac{1+\epsilon}{1-\epsilon}\log{\mathcal{T}},\\
    \nonumber
    & \min_{\boldsymbol{\theta}\in\Big\{ S^{'}_{\mathbf{p}_k^{\mathbf{V}}}(\mathbf{V})-\Theta^{''}\Big\}} \langle \mathbf{N}_{\tau_k^{\mathbf{V}}}-e^{(j)}, \text{KL}_{\boldsymbol{\theta}} \rangle < \frac{1+\epsilon}{1-\epsilon}\log{\mathcal{T}},
\end{align}
where $e^{(j)}$ denotes the j-th unit vector in finite-dimensional Euclidean space. The last two constraints ensure that price $\mathbf{p}_k^{\mathbf{V}}$ is eliminated at time $t_k^{\mathbf{V}}$ (no earlier and no later).
\end{theorem}

\begin{proof}
In Con-TS-RTP with \textit{Constraint Set A}, the aggregator's daily objective and constraints are dependent on the sampled parameter $\boldsymbol{\widetilde{\theta}}_{\tau}$. The only difference between Con-TS-RTP and the daily optimization in \cite{drizzy} is the added constraints. Since the constraints are only enforced for the sampled parameter, each sampled parameter $\boldsymbol{\widetilde{\theta}}_{\tau}$ still has a unique optimal price signal, and more importantly, the constraints do not prohibit the algorithm from selecting the optimal price for the sampled parameter. As such, the addition of constraints that depend only on the daily sampled parameter does not alter the bandit problem, and the regret analysis follows from \cite{drizzy}, which depends heavily on \cite{gopalan2014thompson}.
\end{proof}
\vspace{-10pt}
\subsection{Discussion on Operational Reliability}
\addtocounter{proposition}{-1}
\begin{proposition}
\label{prop}
(Repeated) Under assumptions \ref{assump finite}-\ref{assump 2}, with $\nu$ in equations \eqref{prob1}-\eqref{prob3} chosen such that $\nu\leq\mu\pi_{0}(\boldsymbol{\theta}^{\star})
e^{-\lambda|\mathcal{P}|}$,  with probability $1-\delta\sqrt{2}$, the Con-TS-RTP algorithm with Constraint Set B will uphold the probabilistic distribution system constraints as formulated in \eqref{NOTclairvoyant_constraint} for each day $\tau=1,\dots,\mathcal{T}$. 
\end{proposition}

\vspace{-10pt}
\begin{proof}
In \cite{gopalan2014thompson}, it is shown that with probability $1-\delta\sqrt{2}$ the mass of the true parameter never decreases below $\pi_{0}(\boldsymbol{\theta}^{\star})e^{-\lambda|\mathcal{P}|}$ in the prior distribution during the entire learning process. 
\noindent As such, the desired reliability metric on the RHS of the constraints \eqref{prob1}-\eqref{prob3}, i.e., $1-\nu$, can be selected such that the constraints \textit{must} hold for the true parameter. Let $\pi_{min}^{\star} = \pi_{0}(\boldsymbol{\theta}^{\star})e^{-\lambda|\mathcal{P}|}$ be the minimum reachable mass of the true parameter in the prior distribution. Furthermore, we abuse notation and denote $\mathbb{P}_j^{safe} = \mathbb{P}_{\{\phi_c\}_{c\in\mathcal{C}}}\big[g_j\big(\mathbf{D}_{\tau}^{}(\mathbf{p}_{\tau})\big)\leq0\big]$ as the probability that constraint $j$ is upheld. 
Now, assuming the aggregator only has knowledge of the true parameter given by the prior distribution $\pi_{\tau}$ on day $\tau$, the aggregator can calculate the probability of satisfying the constraint:
\vspace{-2pt}
\begin{align}
    \sum_{\boldsymbol{\hat{\theta}}\in\Theta} \pi_{\tau}(\boldsymbol{\hat{\theta}})(\mathbb{P}_j^{safe}|\boldsymbol{\theta}=\boldsymbol{\hat{\theta}}).
\end{align}

This can be split into two terms for the true parameter $\boldsymbol{\theta}^{\star}$ and all other parameters $\boldsymbol{\theta}\neq\boldsymbol{\theta}^{\star}$:
\begin{align}
   \pi_{\tau}(\boldsymbol{{\theta}^{\star}})(\mathbb{P}_j^{safe}|\boldsymbol{\theta}=\boldsymbol{\theta}^{\star}) + (1-\pi_{\tau}(\boldsymbol{{\theta}^{\star}}))(\mathbb{P}_j^{safe}|\boldsymbol{\theta}\neq\boldsymbol{\theta}^{\star}).
\end{align}
Now, we can rewrite the probability assuming that $\boldsymbol{\theta}^{\star}$ has reached the minimum mass $\pi_{min}^{\star}$ in the prior distribution:
\begin{align}
    \pi_{min}^{\star}(\mathbb{P}_j^{safe}|\boldsymbol{\theta}=\boldsymbol{\theta}^{\star}) + (1-\pi_{min}^{\star})(\mathbb{P}_j^{safe}|\boldsymbol{\theta}\neq\boldsymbol{\theta}^{\star}).
\end{align}
Recall, the aggregator wants constraint $j$ to hold with probability at least $1-\mu$ for the true parameter $\boldsymbol{\theta}^{\star}$, so we can replace $(\mathbb{P}_j^{safe}|\boldsymbol{\theta}=\boldsymbol{\theta}^{\star})$ with $1-\mu$. Furthermore, $(\mathbb{P}_j^{safe}|\boldsymbol{\theta}\neq\boldsymbol{\theta}^{\star})\leq1$ and we replace it accordingly yielding:
\begin{align}
    \pi_{min}^{\star}(1-\mu) + (1-\pi_{min}^{\star}).
\end{align}
Now, we want this probability to be the minimum allowable probability across the prior $\pi$ for constraint $j$ to hold so we set it equal to the reliability metric:
\begin{align}
    \pi_{min}^{\star}(1-\mu) + (1-\pi_{min}^{\star}) = 1-\nu,
\end{align}
which yields
\vspace{-8pt}
\begin{align}
    \nu = \mu\pi_{min}^{\star}.
\end{align}
Moreover, by selecting $\nu \leq \mu\pi_{min}^{\star}$ the aggregator ensures that constraint $j$ will be upheld with probability at least $1-\mu$ for the true parameter $\boldsymbol{\theta}^{\star}$. (i.e., the total mass of the incorrect  parameters $\boldsymbol{\theta}\neq\boldsymbol{\theta}^{\star}$ in the prior distribution $\pi_\tau$ can never be large enough to satisfy the constraint's inequality without the true parameter also satisfying the constraint). \end{proof} 


\end{document}